%% file: main.tex
\documentclass[submission,copyright,creativecommons]{eptcs}
\providecommand{\event}{GandALF 2021} 
\usepackage{underscore}           

\usepackage{times}
\usepackage{soul}
\usepackage{url}
\usepackage[utf8]{inputenc}
\usepackage[small]{caption}
\usepackage{graphicx}
\usepackage{amsmath}
\usepackage{amsthm}
\usepackage{booktabs}
\usepackage{algorithm}
\usepackage{algorithmic}
\usepackage[english]{babel}

\usepackage{booktabs}
\usepackage{algorithm}
\usepackage{algorithmic}
\usepackage{paralist}                                       
\usepackage{nccmath}
\usepackage{xspace}
\usepackage{amssymb}
\usepackage{float}
\usepackage{cmll}
\usepackage{subfig}
\usepackage[shortlabels]{enumitem}
\usepackage{tikz}

\input{macros}

\newcommand{\mona}{{$\mathsf{MONA}$}\xspace}
\newcommand{\spot}{{$\mathsf{SPOT}$}\xspace}

\renewcommand{\LTLf}{{LTL$_f$}\xspace}

\renewcommand{\LTL}{{LTL}\xspace}
\renewcommand{\ltlf}{{LTL$_f$}\xspace}
\renewcommand{\pltlf}{{PLTL$_f$}\xspace}

\renewcommand{\PLTLf}{{PLTL$_f$}\xspace}
\newcommand{\DB}{\mathcal{DB}} 

\newtheorem{lemma}{Lemma}
\newtheorem{definition}{Definition}

\newtheorem{theorem}{Theorem}

\title{On the Power of Automata Minimization in Reactive Synthesis}
\author{Shufang Zhu
\institute{Sapienza Universit\`a di Roma\\ Roma, Italy}
\email{zhu@diag.uniroma1.it}
\and
Lucas M. Tabajara
\institute{Rice University\\ Houston, USA}
\email{lucasmt@rice.edu}
\and
Geguang Pu\thanks{Corresponding author}
\institute{East China Normal University\\ Shanghai, China}
\email{ggpu@sei.ecnu.edu.cn}
\and
Moshe Y. Vardi
\institute{Rice University\\ Houston, USA}
\email{vardi@cs.rice.edu}
}
\def\titlerunning{On the Power of Automata Minimization in Reactive Synthesis}
\def\authorrunning{S. Zhu, L. M. Tabajara, G. Pu \& M. Y. Vardi}
\begin{document}
\maketitle

\input{0-abstract}
\input{1-introduction}
\input{2-preliminaries}

\input{3-brzozowski-alg}

\input{4-experiments}
\input{5-analysis}

\section*{Acknowledgments}
Research partially supported by the ERC Advanced Grant WhiteMech (No. 834228), the EU ICT-48 2020 project TAILOR (No. 952215), National Key R\&D Program of China (2020AAA0107800), NSF grants IIS-1527668, CCF-1704883, IIS-1830549, DoD MURI grant N00014-20-1-2787, and an award from the Maryland Procurement Office. 

\clearpage

\nocite{*}
\bibliographystyle{eptcs}
\bibliography{ref}

\end{document}

%% file: macros.tex
\DeclareMathAlphabet{\mathcal}{OMS}{cmsy}{m}{n}
\newcommand{\A}{\mathcal{A}} \newcommand{\B}{\mathcal{B}}
\newcommand{\C}{\mathcal{C}} \newcommand{\D}{\mathcal{D}}

 \renewcommand{\L}{\mathcal{L}}
 \newcommand{\N}{\mathcal{N}}
 \renewcommand{\P}{\mathcal{P}}
 
\renewcommand{\S}{\mathcal{S}}

 \newcommand{\X}{\mathcal{X}}
\newcommand{\Y}{\mathcal{Y}} \newcommand{\Z}{\mathcal{Z}}

\newcommand{\LTL}{{\sc ltl}\xspace}
\newcommand{\LTLf}{{\sc ltl}$_f$\xspace}
\newcommand{\ltlf}{{\sc ltl}$_f$\xspace}
\newcommand{\pltlf}{{\sc pltl}$_f$\xspace}

\newcommand{\PLTLf}{{\sc pltl}$_f$\xspace}

\newcommand{\NFA}{{NFA}\xspace}

\newcommand{\DFA}{{DFA}\xspace}

\newcommand{\Nat}{{\rm I\kern-.23em N}}


%% file: 0-abstract.tex
\begin{abstract}

Temporal logic is often used to describe temporal properties in AI applications. The most popular language for doing so is Linear Temporal Logic~(\LTL). Recently, \LTL on finite traces, \LTLf, has been investigated in several contexts. In order to reason about \LTLf, formulas are typically compiled into deterministic finite automata~(\DFA), as the intermediate semantic representation. Moreover, due to the fact that DFAs have canonical representation, efficient minimization algorithms can be applied to maximally reduce \DFA size, helping to speed up subsequent computations. Here, we present a thorough investigation on two classical minimization algorithms, namely, the Hopcroft and Brzozowski algorithms.
More specifically, we show how to apply these algorithms to semi-symbolic
(explicit states, symbolic transition functions) automata representation.
We then compare the two algorithms in the context of an \LTLf-synthesis framework, starting from \LTLf formulas. While earlier studies on comparing the two algorithms starting from randomly-generated automata concluded that neither algorithm dominates, our results suggest that starting from \LTLf formulas, Hopcroft's algorithm is the best choice in the context of reactive synthesis. Deeper analysis explains why the supposed advantage of Brzozowski's algorithm does not materialize in practice.
\end{abstract}

%% file: 1-introduction.tex
\section{Introduction}
In many situations in Formal Methods and AI, we are interested in expressing properties over a sequence of successive states. Temporal logic, especially Linear Temporal Logic~(\LTL) has been thoroughly investigated for doing so~\cite{Pnu77}. Recently, a variant of \LTL on finite traces, namely \LTLf, has been investigated~\cite{DegVa13}. \LTLf found application in numerous AI contexts, as it is suitable for expressing properties over an unbounded but finite sequence of successive states. When reasoning about actions and planning, \ltlf has been employed as a specification mechanism 
for finite-horizon temporally extended goals~\cite{CTMBM17,DegRu18}. 
As a specification language, we can use \LTLf to specify desired properties in machine learning~\cite{CamachoIKVM19,GiacomoIFP19,abs-2101-11981}, program synthesis~\cite{DegVa15,ZTLPV17,CamachoBMM18}, Business Process Management (BPM)~\cite{PesicSA07,GiacomoMGMM14,CiccioMMM17}, Markov Decision Processes~(MDPs) with {non}-Markovian rewards~\cite{BrafmanGP18}, MDPs policy synthesis~\cite{AndrewGandalf}, also non-Markovian planning and decision problems~\cite{BrafmanG19}. A general survey of applications of \LTLf in AI and CS can be found in~\cite{DegVa13,GiacomoSFR20}.

In many applications, the common technique for reasoning about \LTLf is compiling formulas into deterministic finite automata~(\DFA), cf.~\cite{ZTLPV17}. Unfortunately, the \DFA size can be, in the worst case, doubly-exponential to the length of the formula~\cite{KupfermanVa01}. Indeed, \LTLf-to-\DFA compilation has been shown to be the bottleneck of \LTLf synthesis~\cite{ZTLPV17}. On the positive side, DFAs can be fully minimized thus helping to speed up subsequent computations \cite{Hopcroft71,Brzozowski62,DegVa15}. Furthermore, there is evidence that the doubly-exponential blow-up of \LTLf-to-\DFA compilation,  does not tend to occur in practice~\cite{TabakovRV12}. This means that applications that require compiling temporal knowledge in \LTLf to DFAs can still be implemented efficiently for many instances, as long as a good minimization algorithm is used. The natural question to ask, then, is: what is the best way in practice of constructing a minimal DFA from an \LTLf formula?

An empirical evaluation of \DFA-minimization algorithms can be found in~\cite{TabakovV05}. That work compares two classical algorithms for constructing a minimal \DFA from an NFA (nondeterministic finite automaton). The first is Hopcroft's algorithm~\cite{Hopcroft71}, which first determinizes the NFA into a (not necessarily minimal) \DFA, and then partitions the state space into equivalence classes. These equivalence classes then correspond to the states of the minimal \DFA. The second is Brzozowski's algorithm~\cite{Brzozowski62}, which reverses the automaton twice, determinizing and removing unreachable states after each reversal. This sequence of operations guarantees that the resulting DFA is minimal. The conclusion reached by~\cite{TabakovV05} is that neither algorithm dominates across the board, and the best algorithm to use for a given NFA depends in part on the NFA's transition density. 


A few aspects make the evaluation in ~\cite{TabakovV05} unsatisfactory for our purposes. First, the algorithms were compared considering an NFA as a starting point, while we are interested in obtaining minimal DFAs from \LTLf formulas. This difference in initial representation may require certain steps of the algorithms to be implemented in a different way that affects their complexity. Second, the evaluation was performed on NFAs generated using a random model, which might not be representative of automata compiled from \LTLf formulas. Third, automata generated from formulas tend to be semi-symbolic, having their transitions represented symbolically by data structures such as Binary Decision Diagrams~(BDDs) \cite{Mona}, a commonly being used representation method that is more compact than the classical explicit representation. This semi-symbolic representation can also affect how certain operations are implemented and therefore the performance of the algorithms. Moreover, studies from~\cite{TabakovV05} started with random automata, with the lack of looking into practical applications. We would like to look deeper into the context of applications. A number of applications require the step of automata minimization. Examples are monitoring~\cite{TabakovRV12}, reactive synthesis~\cite{ZhuTLPV17}, and Business Process Management (BPM)~\cite{PesicSA07,GiacomoMGMM14,CiccioMMM17}. 

In this work we re-examine the comparison between the Hopcroft and Brzozowski algorithms, this time starting from \LTLf formulas. In particular, we focus on the context of \emph{reactive synthesis}, the algorithms of which usually make use of automata-theoretic techniques to automatically convert system properties described in high-level specification into a reactive system satisfying these properties, cf.~\cite{PnueliR89}. In the standard approach for solving this problem for \LTLf, the \LTLf specification is first compiled into a \DFA, and then a system satisfying the specification is synthesized by solving a reachability game over this \DFA~\cite{DegVa15}. The game-solving step is performed over a fully-symbolic representation of the DFA that also encodes the state space symbolically, and which is thus exponentially smaller than the explicit representation~\cite{ZTLPV17}. It has also been shown that synthesis techniques employing minimized DFA often shows dominating performance than the ones without~\cite{TabajaraV19, BLTV}. 

In our evaluation, Hopcroft's algorithm is represented by the tool \mona{}~\cite{Mona}, a sophisticated platform for constructing DFAs from temporal logical specifications, commonly being used for minimal DFA construction in \LTLf synthesis. It should be noted that, though \mona{} is able to handle full \emph{second-order logic}~(MSO), which subsumes \emph{first-order logic}~(FOL), it has been shown that expressing \ltlf{} in FOL gives better performance~\cite{ZhuPV19}. \mona{} constructs a DFA in a bottom-up fashion, first constructing small DFAs for subformulas and then progressively combining them while applying Hopcroft minimization after each step. Since there is no preexisting tool that makes use of Brzozowski's algorithm, we show here how it can be effectively simulated within the existing framework of \ltlf{} synthesis. This is done by using \mona{} to construct a minimal DFA for the \emph{reverse} language of the \ltlf{} specification, which can then be reversed, determinized and pruned of unreachable states to obtain the minimal DFA for the original language. The possible advantage of Brzozowski's algorithm is that the DFA for the reverse language of an \ltlf{} formula is guaranteed to be at most exponential, rather than doubly exponential, in the size of the formula~\cite{CKS81,DegVa13}.

We present two approaches for performing the final determinization step in Brzozowski's algorithm: an explicit approach, using routines implemented in the SPOT automata library~\cite{spot}, and a symbolic approach, which converts the reversed automaton directly into a fully symbolic DFA. The benefit of the fully symbolic approach is that it avoids constructing the semi-symbolic DFA, which may be exponentially larger than its fully-symbolic representation. On the other hand, this also means that the \DFA is converted to the fully-symbolic representation before removing the unreachable states, which can lead to this representation being more complex than necessary. After the fully symbolic representation is constructed, computing the reachable states serves only to reduce the search space during synthesis, but does not reduce the size of the representation. While symbolic determinization has been discussed in prior works~\cite{ArmoniEFKV05,MorgensternS08}, the role that symbolic determinization can play in the framework of \ltlf~\cite{ZTLPV17,BLTV} and the impact it may have on algorithmic performance has not been investigated yet.

We compare Hopcroft's Algorithm and the two versions of Brzozowski's Algorithm on a number of \ltlf{}-synthesis benchmarks
by evaluating not only the performance of the DFA construction, but also how the resulting fully-symbolic DFA affects the end-to-end synthesis performance. We find that, despite the minimal DFA for the reverse language having an exponentially smaller theoretical upper bound compared to the minimal DFA for the original language, in practice it is often a similar size or even larger. As a consequence, this observation suggests that Brzozowski's algorithm does not benefit in practice from performing determinization symbolically, as the fully-symbolic representation becomes much less efficient and the reachability computation does not compensate for the overhead. Finally, we observe that Hopcroft's algorithm dominates significantly in the wide majority of the cases. We thus conclude that unlike in~\cite{TabakovV05}, where the evaluation indicated that both minimization algorithms perform well in different cases, in the context of synthesis Hopcroft's algorithm is likely preferable. Moreover, the discrepancy between theory and practice, which leads to the disappearing of the exponential blow-up from the DFA of the reverse language to the DFA of the original \LTLf formula, suggests the Hopcroft approach as a promising option also in other scenarios that require obtaining minimal \DFA from \LTLf~\footnote{Some proofs are moved to the appendix due to the lack of space. See full version on arXiv.}.

%% file: 2-preliminaries.tex
\section{Preliminaries}\label{sec:pre}
\subsection{\LTLf and Pure-Past \LTLf}
Linear Temporal Logic over finite traces, called \LTLf~\cite{DegVa13} extends propositional logic with finite-horizon temporal connectives. In particular, \LTLf can be considered as a variant of Linear Temporal Logic~(LTL)~\cite{Pnu77}. The key feature of \LTLf is that it is interpreted over finite traces, rather than infinite traces as in LTL. Given a set of propositions $\P$, the syntax of \ltlf is identical to LTL, and defined as:

$$\phi ::= \top\ |\ \bot\ |~p \in \P~|~(\neg \phi) ~| ~ (\phi_1\wedge\phi_2)~ |~ (X\phi)~|~(\phi_1 U \phi_2).$$
	
$\top$ and $\bot$ represent \textit{true} and \textit{false} respectively. $X$~(Next) and $U$~(Until) are temporal connectives. Other temporal connectives can be expressed in terms of those.
A \textit{trace} $\rho = \rho[0],\rho[1],\ldots$ is a sequence of propositional assignments~(sets), in which  $\rho[m]\in 2^\P$ ($m \geq 0$) is the $m$-th assignment of $\rho$, and $|\rho|$ represents the length of $\rho$. Intuitively, $\rho[m]$ is interpreted as the set of propositions that are $true$ at instant $m$. A trace $\rho$ is \textit{infinite} if $|\rho| = \infty$, denoted as $\rho\in (2^\P)^{\omega}$, otherwise $\rho$ is \textit{finite}, denoted as $\rho\in (2^\P)^{*}$. We assume standard temporal semantics from~\cite{DegVa13}, and we write $\rho \models \phi$, if finite trace $\rho$ satisfies $\phi$ at instant $0$. We define the language of a formula $\phi$ as the set of traces satisfying $\phi$, denoted as $\L(\phi)$.

We now introduce \PLTLf, which is the \emph{pure-past} version of \LTLf that considers the past instead of the future~\cite{ZhuPV19,GiacomoSFR20}. 
\PLTLf is defined as follows:

$$\theta ::= \top\ |\ \bot\ |\ p \in \P\ |\ (\neg \theta)\ |\ (\theta_1\wedge\theta_2)\ |\ (Y\theta)\ |\ (\theta_1 S \theta_2).$$

We assume standard temporal semantics from~\cite{ZhuPV19,GiacomoSFR20}. 
\PLTLf has a natural interpretation on finite traces: the formula is satisfied if it holds in the last instant of a trace. Consider finite trace $\rho$, we say that $\rho \models \theta$, if $\rho,k-1\models\theta$, where $k=|\rho|$. 
We define the language $\L(\theta)$ as the set of finite traces satisfying $\theta$, that is, $\L(\theta) = \{\rho \mid \rho \models \theta\}$.

Consider \LTLf formula $\phi$, we can reverse it by replacing each temporal connective in $\phi$ with the corresponding \emph{past} connective from \PLTLf thus getting $\phi^R$. $X$ and $U$ are replaced by $Y$ and $S$, respectively. For a finite trace $\rho$, we define $\rho^R = \rho[|\rho|-1],\rho[|\rho|-2],\ldots,\rho[1],\rho[0]$ to be the reverse of  $\rho$.  We define the reverse of $\mathcal{L}$ as the set of reversed traces in $\mathcal{L}$, denoted as $\mathcal{L}^R$; formally, $\mathcal{L}^R = \{\rho^R \mid \rho \in \mathcal{L}\}$. The following theorem shows that \PLTLf formula $\phi^R$ accepts exactly the reverse language of $\phi$.
\begin{theorem}\label{thm:pltlf}{\rm \cite{ZhuPV19}}
Let $\mathcal{L}(\phi)$ be the language of \LTLf formula $\phi$ and $\mathcal{L}^R(\phi)$ be the reverse language, then $\mathcal{L}^R(\phi) = \mathcal{L}(\phi^R)$.
\end{theorem}

\subsection{Automata Representations} \label{sec:representations}
An \LTLf formula can be compiled into an automaton over finite words that accepts a trace if and only if that trace satisfies the formula. Here we introduce a few different automata representations. The difference between the representations here and the standard textbook representation of finite-state automata is that the alphabet is defined in terms of truth assignments to propositions.


\begin{definition}[Nondeterministic Finite Automata]
	An NFA is represented as a tuple $\N = (\P, \S, \S_0, \eta, Acc)$, where
	\begin{inparaitem}
		\item $\P$ is a finite set of propositions;
		\item $\S$ is a finite set of states;
		\item $\S_0 \subseteq \S$ is a set of initial states;
		\item $\eta: \S \times 2^\P \rightarrow 2^\S$ is the transition function such that given current state $s \in \S$ and an assignment $\sigma \in 2^\P$, $\eta$ returns a set of successor states;
		\item $Acc \subseteq \S$ is the set of accepting states.
	\end{inparaitem}
\end{definition}
If there is only one initial state $s_0$ and $\eta$ returns a unique successor for each state $s \in \S$ and assignment $\sigma \in 2^\P$, then we say that $\N$ is a deterministic finite automaton~(DFA). In this case, $\eta$ can be written in the form of $\eta : \S \times 2^\P \rightarrow \S$. The set of traces accepted by $\N$ is called the \emph{language} of $\N$ and denoted by $\L(\N)$. Moreover, we also introduce here a so-called codeterministic finite automaton~(co-DFA)~\cite{Pin87}. $\N$ is called a co-DFA if for each state $s \in \S$ and transition condition $\sigma \in 2^\P$, there is a unique predecessor $d$ such that $\eta(d, \sigma) = s$. Intuitively, an NFA with a unique accepting state $Acc = \{s_{acc}\}$ is considered as a co-DFA if reversing all the transitions and switching $\S_0$ with $Acc$ gives us a DFA. 


A question that remains is how to represent the transition function $\eta$ efficiently. It could be represented by a table mapping states and assignments in $2^\P$ to the set of successor states, but this table would necessarily be exponential in the number of propositions. In practice, from a given state it is usually the case that multiple assignments can lead to a same successor state. These assignments can then be represented collectively by a single Boolean formula $\lambda$. For a given state, the number of such formulas is usually much smaller than the number of assignments.
Therefore, the transition function can alternatively be represented by a relation $H : \S \times \Lambda \times \S$, where $\Lambda$ is a set of propositional formulas over $\P$. We then have $(s, \lambda, d) \in H$ for a formula $\lambda$, iff $d \in \eta(s, \sigma)$ for every $\sigma \in 2^\P$ that satisfies $\lambda$. Intuitively, the tuples of $H$ can be thought of as edges in the graph representation of the automaton, labeled by the propositional formulas that match the transitions. It should be noted that \mona~\cite{Mona} adopts this representation, representing propositional formulas as Binary Decision Diagrams~(BDDs)~\cite{Bryant92}.

We call the above a \emph{semi-symbolic} automaton representation, as transitions are represented symbolically by propositional formulas but the states are still represented explicitly. In contrast, we now present a \emph{fully-symbolic}~(\emph{symbolic} for short) representation, in which both states and transitions are represented symbolically. In the fully-symbolic representation, states are encoded using a set of \emph{state variables} $\Z$, where a state corresponds to an assignment of $\Z$.

\begin{definition}[Symbolic Deterministic Finite Automaton]\label{def:symbolic_dfa}
	A symbolic \DFA of a corresponding explicit \DFA $\A = (\P, \S, s_0, \eta, Acc)$, in which $\eta$ is in the form of $\eta: \S \times 2^\P \rightarrow \S$, is represented as a tuple $\D = (\P, \Z, I, \delta, f)$, where 
	\begin{itemize}[noitemsep]
		\item $\P$ is the set of propositions as in $\A$;
		\item $\Z$ is a set of state variables with $|\Z| = \lceil \log |\S| \rceil$, and every state $s$ in the explicit DFA corresponds to an assignment $Z \in 2^\Z$ of propositions in $\Z$;
		\item $I \in 2^\Z$ is the initial assignment corresponding to $s_0$;
		\item $\delta: 2^\Z \times 2^\P \rightarrow 2^\Z$ 
		is the transition function. Given assignment $Z$ of current state $s$ and transition condition $\sigma$, $\delta(Z, \sigma)$ returns the assignment $Z'$ corresponding to the successor state $s' = \eta(s,\sigma)$;
		\item $f$ is a propositional formula over $\Z$ describing the accepting states, that is, each satisfying assignment $Z$ of $f$ corresponds to an accepting state $s \in Acc$.
	\end{itemize}
\end{definition}

Note that the transition function $\delta$ can be represented by an indexed family consisting of a Boolean formula $\delta_z$ for each state variable $z \in \Z$, which when evaluated over an assignment to $\Z \cup \P$ returns the next assignment to $z$. Since the states are encoded into a logarithmic number of state variables, depending on the structure of these formulas, the symbolic representation can be exponentially smaller than the semi-symbolic representation. 

%
%

\subsection{Minimized DFA from NFA}\label{sec:dfaminimization}
For every NFA, there exists a unique smallest DFA that recognizes the same language, called the \emph{canonical} or \emph{minimal} DFA. A typical way to construct the canonical DFA for a given NFA is to determinize the automaton using subset construction~\cite{RS59} and then minimize it using, for example, Hopcroft's DFA minimization algorithm~\cite{Hopcroft71}. The idea of Hopcroft's algorithm is as follows. Consider state $s$ in automaton $\N$, we define $\L(s)$ as the set of accepting words of $\N$ having $s$ as the initial state. Note that a minimal DFA cannot have two different states such that $\L(s) = \L(s')$. Hopcroft's algorithm computes equivalence classes of states, such that two states $s$ and $s'$ are considered equivalent if $\L(s) = \L(s')$. 
\begin{theorem}{\rm\cite{Hopcroft71}} \label{thm:hopcroft}
Let $\N$ be an NFA. Then {\centering$\A = [equivalence \circ determinize](\N)$} is the minimal DFA accepting the same language as $\N$.
\end{theorem}
Fcuntion $determinize(\N)$ is the DFA obtained by applying subset construction to $\N$, and function $equivalence(\N)$ partitions the set of states into equivalence classes, which then form the states in the canonical DFA. The initial partition is $Acc$ and $\S \backslash Acc$.
Then at each iteration this partition is refined by splitting each equivalence class, until no longer possible.
%
\mona{}~\cite{Mona}, the state-of-the-art practical tool for constructing minimal DFA from logic specifications~\cite{Mona,ZTLPV17}, operates by induction on the structure of the input formulas, constructing DFAs for the subformulas and then combining them recursively, while applying Hopcroft's minimization algorithm after each step. We thus say that \mona{} follows the Hopcroft approach for minimization, through an optimized variant adapted to the semi-symbolic automata used by \mona{}.


A less direct way to construct minimal DFAs from NFAs, in which there is no explicit step of minimization, is due to Brzozowski~\cite{Brzozowski62}. We use the following formulation of Brzozowski's approach~\cite{TabakovV05}.  For notation, $reverse(\N)$ is the function that maps the NFA $\N = (\P, \S, \S_0, \eta, Acc)$ to the NFA $\N^R = (\P, \S, Acc, \eta^R, \S_0)$, where $(d, \sigma, s) \in \eta^R$ iff $(s, \sigma, d) \in \eta$; $determinize(\N)$ again returns the DFA by applying subset construction to $\N$; and $reachable(\N)$ is the automaton resulting from removing all states that are not reachable from the initial states of $\N$.

\begin{theorem}{\rm\cite{Brzozowski62}} \label{thm:brzozowski}
	Let $\N$ be an NFA. Then {\centering$\A = [reachable \circ determinize \circ reverse]^2(\N)$} is the minimal DFA accepting the same language as $\N$.
\end{theorem}



\subsection{Symbolic \LTLf synthesis}\label{sec:symbolicltlf}
\begin{definition}[\LTLf Synthesis]
	Let $\phi$ be an \LTLf formula over $\P$ and $\X, \Y$ be two disjoint sets of propositions such that $\X\cup \Y = \P$. $\X$ is the set of \emph{input variables} and $\Y$ is the set of \emph{output variables}. $\phi$ is \emph{realizable} with respect to $\langle\X, \Y\rangle$ if there exists a strategy $g: (2^{\X})^*\rightarrow 2^{\Y}$, such that for an arbitrary infinite sequence $\pi = X_0,X_1,\ldots\in (2^{\X})^{\omega}$ of propositional assignments over $\X$, we can find $k\geq 0$ such that $\phi$ is $true$ in the finite trace $\rho = (X_0\cup g(\epsilon)), (X_1\cup g(X_0)),\ldots , (X_k\cup g(X_0,X_1,\ldots,X_{k-1}))$.
\end{definition}
Intuitively, \LTLf synthesis can be thought of as a game between two players: the \emph{environment}, which controls the input variables, and the \emph{agent}, which controls the output variables. Solving the synthesis problem means synthesizing a strategy $g$ for the agent such that no matter how the environment behaves, the combined behavior trace of both players satisfies the logical specification $\phi$~\cite{DegVa15}. There are two versions of the synthesis problem depending on which player acts first. Here we consider the agent as the first player, but a version where the environment moves first can be obtained with small modifications. In both versions, however, the agent decides when to end the game.

The state-of-the-art approach to \LTLf synthesis is the symbolic approach proposed in~\cite{ZTLPV17}. This approach first translates the \LTLf specification to a first-order formula, which is then fed to \mona to get the fully-minimized semi-symbolic DFA. This DFA is transformed to a fully-symbolic DFA, using BDDs~\cite{Bryant92} to represent each $\delta_z$ as well as the formula $f$ for the accepting states.
Solving a reachability game over this symbolic DFA settles the original synthesis problem. The game is solved by performing a least fixpoint computation over two Boolean formulas $w$ over $\Z$ and $t$ over $\Z \cup \Y$, which represent the set of all winning states and all pairs of winning states with winning outputs, respectively. 
Intuitively, winning states are those from which the agent has a winning strategy, and each winning state $Z$ has a winning output $Y$ that refers to the agent action returned by the winning strategy.
$t$ and $w$ are initialized as $t_0(Z, Y) = f(Z)$ and $w_0(Z) = f(Z)$, since every accepting state is an agent  winning state. 
Then $t_{i+1}$ and $w_{i+1}$ are constructed as follows:
\begin{itemize}[noitemsep]
	\item 
	$t_{i+1}(Z, Y) = t_i(Z, Y) \lor (\neg w_i(Z) \land \forall X . w_i(\delta(X, Y, Z)))$
	\item 
	$w_{i+1}(Z) = \exists Y . t_{i+1}(Z, Y)$
\end{itemize}

The computation reaches a fixpoint when $w_{i+1} \equiv w_i$. At this point, no more states will be added, and so all agent winning states have been collected. By evaluating $w_i$ on $I$ we can know if there exists a winning strategy. If that is the case, $t_i$ can be used to compute a winning strategy. This can be done through the mechanism of Boolean synthesis~\cite{DrLu.cav16}. 
We note that extensions of {\ltlf} synthesis were studied in~\cite{ZhuTLPV17} and~\cite{ZhuGPV20}. 
In all of these works, compiling {\ltlf} formulas as corresponding DFAs proved to be the computational bottleneck.

%% file: 3-brzozowski-alg.tex
\section{Brzozowski's Algorithm from \ltlf{}}\label{sec:brzozowski}
As mentioned in Section~\ref{sec:dfaminimization}, a variation of Hopcroft's algorithm is already implemented in the tool \mona{}, which is the standard tool employed in \ltlf{} synthesis applications. In contrast, there is no existing tool that directly implements Brzozowski's algorithm for temporal specifications. In this section, we describe how the algorithm can be adapted to compile an \ltlf formula into a minimal DFA, presenting both an explicit and symbolic version of the algorithm. This section focuses on the theory of the algorithm; implementation details can be found in Section~\ref{sec:implementation}.


Theorem~\ref{thm:brzozowski} in Section~\ref{sec:dfaminimization} describes how to obtain a minimal \DFA from an \NFA, here we start instead from an \ltlf{} formula. This leads to the following sequence of operations:
\begin{enumerate}
    \item \textbf{Reverse DFA construction:} Construct a minimal DFA that recognizes the reverse of the language of $\phi$. This corresponds to the first round of $reachable \circ determinize \circ reverse$.
    \item \textbf{Reversal into a co-DFA:} Reverse the DFA for the reverse language into a co-DFA for the original language. This corresponds to the $reverse$ operation in the second round.
    \item \textbf{Determinization and pruning:} The last two steps, corresponding to the $determinize$ and $reachable$ operations, can be performed either explicitly or symbolically.
    \begin{enumerate}
        \item \textbf{Explicit:} Apply subset construction to the co-DFA to obtain an explicit DFA, removing states that are not reachable from the initial states.
        \item \textbf{Symbolic:} Convert the explicit co-DFA into a symbolic DFA (defined in Section~\ref{sec:representations}). Next, compute a symbolic representation of the set of reachable states of the symbolic DFA. Since removing states cannot be easily done in the symbolic representation, the symbolic set of reachable states is instead used later to prune the search space during the game-solving step.
    \end{enumerate}
\end{enumerate}
%
%
%
\subsection{Reverse DFA Construction} \label{sec:reversedfa}
Starting from an \ltlf{} formula $\phi$, we first produce a minimal DFA for the \emph{reverse} language of $\phi$. Note that, being minimal, this DFA has no unreachable states. It thus corresponds to a DFA obtained by applying the first round of operations $reachable \circ determinize \circ reverse$. Although minimality is a stronger condition than reachability, having the DFA be minimal improves the performance of future steps.

To construct such a DFA, we use the technique introduced in~\cite{ZhuPV19}. We first convert \ltlf{} formula $\phi$ into a \pltlf{} formula $\phi^R$ for the reverse language. This can be done by simply replacing every temporal connective in $\phi$ with its corresponding past connective. Since \pltlf{} can be translated to first-order logic~\cite{ZhuPV19}, $\phi^R$ can be converted into a DFA and minimized, for example using Hopcroft's algorithm.

It might seem odd to generate the minimal DFA for $\phi^R$ as an intermediate step in a minimization algorithm, when one could simply directly generate the minimal DFA for the original formula $\phi$. The difference, however, is that the minimal DFA for the $\phi^R$ is guaranteed to have size at most exponential in the size of the formula~\cite{CKS81,DegVa13,GiacomoSFR20}, while the DFA for $\phi$ itself can be doubly-exponential~\cite{KupfermanVa01}. Specifically, it is shown in~\cite{DegVa13} how to convert an \ltlf{} formula to an alternating word automaton with a linear number of states, and it is shown in~\cite{CKS81} how to convert an alternating word automaton to a DFA for the reverse language with exponential state blow-up. Therefore, constructing the DFA for the reverse language is, in theory, exponentially more efficient than the direct construction.

\begin{theorem} \label{thm:brzozowski-ltlf}
	Let $\phi$ be an \ltlf{} formula, $\A$ be the minimal DFA for $\phi^R$. Then $\A' = [reachable \circ determinize \circ reverse](\A)$ is the minimal DFA accepting the same language as $\phi$.
\end{theorem}
\begin{proof}
	Since $\A$ is the minimal DFA for $\phi^R$, $\L(\A) = \L(\phi^R)$ holds. Moreover, $\phi^R$ accepts the reverse language of $\phi$ s.t. $\L(\phi^R) = \L^R(\phi)$, so we have $\L(\A) = \L^R(\phi)$. As stated in Theorem~\ref{thm:brzozowski}, the first round of $reachable \circ determinize \circ reverse$ returns a deterministic automaton that contains only reachable states and accepts the reverse language of the original input. Note all of these properties hold for $\A$. Specifically, $\A$ is a minimal DFA, and thus $\A$ is deterministic and contains only reachable states. Also, $\L(\A) = \L^R(\phi)$, so $\A$ accepts the reverse language of the original input $\phi$. Therefore, after applying the second round of $reachable \circ determinize \circ reverse$ over $\A$, we get $\A'$, the minimal DFA of $\phi$.
\end{proof}
Theorem~\ref{thm:brzozowski-ltlf} states that the construction of the minimal DFA for the reverse language is able to achieve the first round of $reachable \circ determinize \circ reverse$ in Theorem~\ref{thm:brzozowski}. In the remainder of this section we describe how to perform the second round of operations.
\subsection{Reversal into a Co-DFA}\label{sec:reverse}
The first step produces the minimal DFA $\A$ for the reverse language of $\phi$. 
Reversing a semi-symbolic co-DFA can be done easily in linear time, by only swapping initial states with final states and swapping source with destination for every transition. Note that all transition conditions do not need to be changed.
The result is a co-DFA for the original language $\phi$. As explained in Section~\ref{sec:representations}, 
a co-DFA is a special case of an NFA in which there is only a single transition into a state for each assignment. 

More formally, let $\A = (\P, \S, \S_0, H, Acc)$ be the semi-symbolic DFA for the reverse language, with the (deterministic) transition relation given as $H \subseteq \S \times \Lambda \times \S$, 
where $\Lambda$ is a set of propositional formulas, as described in Section~\ref{sec:representations}. This representation of the transition relation is easy to obtain from the output of \mona{}. 
Reversing $\A$ produces the co-DFA $\C = (\P, \S, Acc, H^R, \S_0)$, where
{\centering$H^R = \{ (d, \lambda, s) \mid (s, \lambda, d) \in H \}.$} Since co-DFA is a special case of NFA, for simplicity, later we still use NFA to refer to this co-DFA.

\subsection{Explicit Minimal DFA Construction} \label{sec:explicit_construction}
The standard way of determinizing an NFA is using subset construction. This construction can be performed in the semi-symbolic representation, with explicit states and symbolic transitions. In this case, each state in the resulting DFA represents a subset of the states in the NFA, and a transition between two states representing subsets $S_1$ and $S_2$ is labeled by the disjunction of all labels $\lambda$ such that $(s_1, \lambda, s_2) \in H$ for $s_1 \in S_1$ and $s_2 \in S_2$. In this semi-symbolic representation, finding the reachable states can be performed by a simple graph search on the graph of the automaton.

The problem with this explicit-state approach is that the subset construction causes an exponential blowup in the state space. This blowup can nullify the advantage obtained by constructing an exponential DFA for $\phi^R$ rather than a doubly-exponential DFA for $\phi$. In the next section we describe how this problem may be mitigated by instead directly constructing a fully-symbolic representation of the DFA.

\subsection{Symbolic Minimal DFA Construction}\label{sec:symbolic_construction}
As described in Section~\ref{sec:symbolicltlf}, the state-of-the-art approach for solving \ltlf{} synthesis uses a fully-symbolic representation of the DFA, which as noted in Section~\ref{sec:representations} can be exponentially smaller. Therefore, the exponential blowup caused by the explicit subset construction described in Section~\ref{sec:explicit_construction} above might be canceled out when the DFA is made fully-symbolic. Constructing the explicit DFA, then, seems like a waste that could be prevented by directly obtaining a symbolic DFA from the semi-symbolic co-DFA. With that in mind, in this section we describe an alternative approach that performs the subset construction and pruning of unreachable states symbolically.

\subsubsection{Symbolic Subset Construction} \label{sec:determinization}

The intuition of the symbolic determinization procedure is that, after subset construction, each state in the \DFA corresponds to a set of \NFA states. Each \DFA state can therefore be represented by an assignment to a set of Boolean variables, one for each \NFA state, where the variable is set to \emph{true} if the corresponding state is in the set. This corresponds naturally to a symbolic representation $\D$ where each explicit state of the \NFA is a state variable in $\D$. 
Therefore, the set of NFA states $\S$ is overloaded as the set of state variables in $\D$. Moreover, in addition to denoting a set of NFA states, $S$ here is also used to denote a DFA state. In this way, $\S$ is able to encode the entire state space of $\D$. 
This approach is reminiscent of the symbolic determinization construction studied in \cite{ArmoniEFKV05} in the context of SAT-based safety LTL model checking, except that our symbolic approach here is BDD-based. (The symbolic determinization construction in \cite{MorgensternS08} is in the context of B\"uchi and co-B\"uchi automata.)

Recall that the transition function $\delta: 2^\S \times 2^\P \rightarrow 2^\S$ in $\D$ can be represented as an indexed family $\{\delta_s \mid 2^\S \times 2^\P \rightarrow \{0,1\} \mid s\in\S\}$. Intuitively speaking, given current \DFA state $S \in 2^\S$ and transition condition $\sigma \in 2^\P$, $\delta_s$ indicates whether state variable $s$ is assigned as $true$ or $false$ in the successor state. Variable $s$ is assigned as $true$ if there is a transition in the \NFA from a state in $S$ that leads to $s$ under transition condition $\sigma$, and $false$ otherwise. 

Therefore, given an NFA $\N = (\P, \S, \S_0, H, Acc)$, the symbolic determinization for the symbolic \DFA $\D = (\P,\S,I,\delta,f)$ proceeds as follows:
\begin{itemize}
	\item 
	$\S$ is the set of state variables;
	\item 
	$I \in 2^\S$ is such that $I(s)=1$ if and only if $s \in \S_0$;
	\item $f = \bigvee_{s\in Acc} s$.
	\item 
	$\delta_s: 2^\S \times 2^\P \rightarrow \{0,1\}$ is such that $\delta_s(S, \sigma) = 1$ iff $(d,\lambda,s) \in H$ for some $d$ such that $S(d) = 1$ and $\lambda$ such that $\sigma \models \lambda$. Each $\delta_s$ can be represented by a formula (or BDD) $\delta_s = \bigvee_{(d, \lambda, s) \in H} (d \land \lambda)$, with $d$ interpreted as a state variable.
\end{itemize}
In order to show that the symbolic determinization described above is correct, i.e., that $\L(\D) = \L(\N)$, we need to prove that the state where the \DFA $\D$ reaches after reading a trace $\rho$ corresponds exactly to the set of states where the \NFA $\N$ can reach after reading $\rho$, which follows the standard subset construction. In the following, we use $\delta(S, \rho)$ to denote the \DFA state that is reached from $S$ by reading $\rho$. Likewise, $H(S, \rho)$ denotes the set $S' \subseteq \S$ of all \NFA states that can be reached from all states $s \in S$ by reading $\rho$.
\begin{lemma}\label{lem}
Let $\rho \in (2^\P)^*$ be a finite trace. \DFA state $\delta(I,\rho)$ encodes the set of \NFA states $H(\S_0,\rho)$, that is, $\{s \mid \delta(I,\rho)(s) = 1\} = H(\S_0,\rho)$.
\end{lemma}

%
%

The following theorem follows directly from Lemma~\ref{lem}.
\begin{theorem}
	$\D$ is equivalent to $\N$, that is, $\L(\D) = \L(\N)$.
\end{theorem}

%

The following theorem states the computational complexity of the symbolic determinization described above. The limiting factor is the construction of $\delta_s$, each of which takes linear time. Therefore, the total complexity is quadratic.
\begin{theorem}
	The symbolic determinization can be done in quadratic time on the size of the NFA.
\end{theorem}

\subsubsection{Symbolic State-Space Pruning}
The symbolic subset construction described in the previous subsection allows us to obtain a DFA for $\phi$ directly in symbolic representation. Although this representation is more compact, it presents further challenges for the final step of pruning unreachable states. This is because in the symbolic DFA the state space is fixed by the set of state variables. Since states are not represented explicitly, but rather implicitly by assignments over these variables, there is no easy way to remove states.

The alternative that we propose is to instead compute a symbolic representation of the set of reachable states, which can then be used during game-solving to restrict the search for a winning strategy. This means that the state space of the automaton, implicitly represented by the state variables, is not minimized, but during game-solving the additional, unreachable states are ignored.

We denote by $r(Z)$ the Boolean formula, over the set of state variables $\Z$, that is satisfied by an assignment $Z \in 2^\Z$ if $Z$ encodes a state that is reachable from the initial state. We compute $r(Z)$ for the symbolic DFA $\D = (\P,\Z,I,\delta,f)$ by iterating the following recurrence until a fixpoint:
\begin{itemize}
	\item 
	$r_0(Z) = I(Z)$
	\item 
	$r_{i+1}(Z) = r_i(Z) \lor \exists Z' . \exists X . \exists Y . r_i(Z') \land (\delta(Z', X \cup Y) = Z)$
\end{itemize}

Once a fixpoint is reached, the resulting formula $r(Z)$ denotes the set of reachable states of the automaton. Then, during the computation of the winning states as described in Section~\ref{sec:symbolicltlf}, we restrict $t_i$ after each step to only those values of $Z$ that correspond to reachable states.

%% file: 4-experiments.tex
\section{Implementation and Evaluation} \label{sec:experiments} 
As mentioned in Section~\ref{sec:dfaminimization}, Hopcroft's algorithm is represented by \mona{}~\cite{Mona}, a sophisticated platform for obtaining minimized automata from logic specifications, in the way of constructing a minimal DFA by first generating DFAs for subformulas, then combining them recursively while applying Hopcroft's algorithm on the intermediate DFAs. 
For more details of \mona{}, we refer to~\cite{Mona}. 
In this section, we first present details of our implementation of Brzozowski's algorithm, and then show an experimental comparison between the two different minimization algorithms.

\subsection{Implementation} \label{sec:implementation}
As shown in Section~\ref{sec:brzozowski}, Brzozowski's algorithm consists of three steps: 1) reverse DFA construction, 2) reversal into a co-DFA, and 3) determinization and pruning. 
For the first step, we translate the \PLTLf formula $\phi^R$ into a first-order formula $fol(\phi^R)$ following the translation in~\cite{ZhuPV19} and then use \mona{} to construct the DFA for $\phi^R$. Since DFAs returned by \mona{} are always minimal, the reverse DFA constructed in this step is guaranteed to be at most exponential in the size of \LTLf formula $\phi$.
Reversing this DFA into a co-DFA for $\phi$ is straightforward. Instead of implementing the reversal step as a separate operation, we optimize by performing the reversal while determinizing. There are two different versions of the determinization and pruning step: explicit and symbolic. We now elaborate on them. 
Note that each version starts with the DFA of $\phi^R$, and we combine the reversal and the subsequent operations of subset construction followed by state-space pruning. 

\subsubsection{Explicit Minimal DFA Construction}
Inspired by~\cite{BLTV}, we borrow the rich APIs from \spot{}~\cite{spot}, a well-developed platform for automata manipulation, to perform each computation step, subset construction in particular. It should be noted that, \spot{} adopts the semi-symbolic representation for automata, where the states are explicit and transitions are symbolic, and therefore, we use it to implement the explicit approach. Note that \spot{} is a platform for $\omega$-automata~(automata over infinite words) manipulation. Therefore we represent the co-DFA as a weak B\"uchi Automaton~(wBA)~\cite{DaxEK07}.

Since the reversal step is straightforward, to simplify the description we consider the co-DFA as the starting point to better show the implementation details. The transformation from co-DFA to wBA follows the techniques presented in~\cite{BLTV}. Intuitively, the technique of transforming to wBA is similar to the translation from \LTLf to \LTL~\cite{DegVa13}. Note that in the translation from \LTLf to \LTL, a fresh proposition $alive \notin \P$ is introduced and required to stay $true$ until the \LTLf formula is satisfied and then stay $false$ forever. In the transformation from co-DFA to wBA, we again use the same proposition $alive$ and introduce a $sink$ state that is triggered by the first moment of  $alive$ being false, such that the finite trace is accepted. Moreover, this $sink$ state is considered as the unique accepting state in the wBA to make sure that $alive$ stays $false$ forever. Formally, if a co-\DFA accepts language $\L(\C)$, then its wBA accepts infinite words in $\{(\rho \wedge alive) \cdot(\neg alive)^\omega  \mid \rho \in \L(\C) \}$, where $\rho \wedge alive$ denotes that $alive$ holds at each instant of finite trace $\rho$. Given co-\DFA $\C = (\P, \S, \S_0, H, Acc)$, we construct the wBA as follows:
\begin{inparaenum}[1)]
	\item 
	\item 
	introduce an extra state $sink$;
	\item 
	for each accepting state $s$ in $Acc$, add a transition from $s$ to $sink$, with transition condition $\neg alive$;
	\item
	for each transition in between states in $\C$, change the transition condition $\lambda $ to $\lambda \wedge alive$;
	\item
	add a self-loop for state $sink$ on $\neg alive$;
	\item
	assign $sink$ as the unique accepting state.
\end{inparaenum}
\begin{theorem}
	Let $\C$ be a co-\DFA, and $\B$ be the wBA generated from the construction (1)-(5) above. 
	Then we have $\L(\B) =  \{(\rho \wedge alive) \cdot(\neg alive)^\omega  \mid \rho \in \L(\C) \}$.
\end{theorem}
%

Then, we are able to use SPOT APIs for wBA to conduct subset construction and unreachable states pruning, thus obtaining the wDBA $\DB$. The corresponding API functions are $\mathsf{tgba\_powerset()}$ and $\mathsf{purge\_unreachable\_states()}$, respectively. Finally, we convert the wDBA $\DB$ back to \DFA as follows:
\begin{inparaenum}[a)]
	\item 
	remove the unique accepting state $\{sink\}$ and transitions leading to or coming from $\{sink\}$;
	\item 
	eliminate $alive$ on each transition by assigning it to $true$;
	\item 
	assign all the states that move to $\{sink\}$ with transition condition $\neg alive$ as the accepting states of the \DFA.
\end{inparaenum}

\begin{theorem}
	Let $\C$ be a co-\DFA, $\B$ the corresponding wBA, $\DB$ the wDBA from SPOT, and $\A$ be the \DFA obtained from the construction (a)-(c) above.
	Then $\A$ is minimal.
\end{theorem}
%

\subsubsection{Symbolic Minimal DFA Construction}\label{sec:symbolic_brzozowski_implementation}
Instead of having each DFA state explicit, the symbolic approach described in Section~\ref{sec:symbolic_construction} is able to directly construct a symbolic DFA as in Definition~\ref{def:symbolic_dfa}. Here, we follow the representation technique used by~\cite{ZTLPV17}, where the propositional formulas for the transition function and accepting states are represented by Binary Decision Diagrams~(BDDs). Moreover, the reversal step is again combined with the symbolic subset construction operation into one step. 
Thus, for state variable $s$, in order to construct BDD for formula $\delta_s: 2^\S \times 2^\P \rightarrow \{0,1\}$, we have to take care of switching the current and successor states of a given transition from the reverse DFA. The same for exchanging BDDs of initial and accepting states.

As for the state-space pruning, formulas represented as BDDs allow us to perform the usual Boolean operations, such as conjunction, disjunction and quantifier elimination. After obtaining the set of reachable states $r(Z)$, during the computation of the winning states as described in Section~\ref{sec:symbolicltlf}, we need to restrict $t_i$ after each step to only those values of $Z$ that correspond to reachable states. There are two ways in which we can do this. The most obvious way is to simply take the conjunction of $t_i(Z, Y)$ with $r(Z)$, thus removing all assignments that correspond to unreachable states. The second option, since we are using BDDs to represent the Boolean formulas $t_i$ and $r$, is to apply the standard BDD $\mathsf{Restrict}$ operation~\cite{cudd}. The BDD produced by $\mathsf{Restrict}(t_i, r)$ still returns $1$ for all satisfying assignments to $t_i$ that also satisfy $r$. Those satisfying assignments that do not satisfy $r$, however, are selectively mapped to either $1$ or $0$, using heuristics to try to make a choice that will lead to a smaller BDD. This corresponds to essentially choosing to keep a subset of the unreachable states if that will lead to a smaller symbolic representation of the set of winning states.

Note that the predecessor computation used to compute $t_{i+1}$ may add unreachable states, therefore it is necessary to apply the conjunction or restriction operation at every iteration, rather than just once. 

\subsection{Experimental Evaluation}
In order to evaluate the Hopcroft's and Brzozowski's minimization algorithms starting from \LTLf formulas, we focus on
the context of temporal synthesis. To do so, we conducted extensive experiments over different classes of benchmarks curated from prior works, spanning classes of realistic and synthetic benchmarks~\cite{ZTLPV17,TabajaraV19,BLTV}.
The benchmarks consist of two classes. The first class of benchmarks is the \emph{Random} family, composed of 1000 \LTLf formulas formed by random conjunction, generated as described in~\cite{ZTLPV17}.
The second one is from~\cite{TabajaraV19,BLTV}, and describes two-player games, split into the \emph{Single-Counter}, \emph{Double-Counters} and \emph{Nim} benchmark families. Here we assign the agent as the first-player. It should be noted that, although different player order might lead to different realizability result, the automata minimization performance, nevertheless, stays the same.

The results shown here represent the end-to-end execution of the synthesis algorithms, from an \LTLf specification to a winning strategy. Therefore, they include both the time for \LTLf-to-DFA compilation and game solving. All tests were run on a computer cluster with exclusive access to a node with Intel(R) Xeon(R) CPU E5-2650 v2 processors.
Timeout was 1000 seconds and memory out was 8G.

Since there are two ways of performing Brzozowski's algorithm, as presented in Section~\ref{sec:brzozowski}, we have in total 3 different approaches, namely Hopcroft, Explicit-Brzozowski and Symbolic-Brzozowski. As introduced in Section~\ref{sec:symbolic_brzozowski_implementation}, during symbolic state-space pruning we are able to either apply restriction or conjunction to access only reachable states during game-solving. We show only the results using restriction, as the difference is not significant and in most cases restriction gives slightly better results.

Figure~\ref{fig:synthesis-random} shows a cactus plot~\footnote{Figures best viewed on a computer.} 
comparing
how the three different approaches perform
on \emph{Random} benchmarks. The curves show how many instances can be solved with a given timeout. The further up the curve is, the more benchmarks could be solved in less time. The graph shows that Hopcroft's minimization algorithm is in fact able to solve many more cases than both versions of Brzozowski's algorithm. Furthermore, the explicit Brzozowski approach slightly outperforms the symbolic one. In spite that with time limit of 10 seconds the symbolic version is able to handle more cases than the explicit one, if we take 1000 seconds as the time limit the explicit version is able to handle in total more cases than the symbolic one. This shows that the explicit version is more scalable than the symbolic one. Increasing the time limit does not change the results, since unsolved instances reached the memory limit.

    \begin{figure}[t]
    	\centering
    	\includegraphics[width=0.6\linewidth]{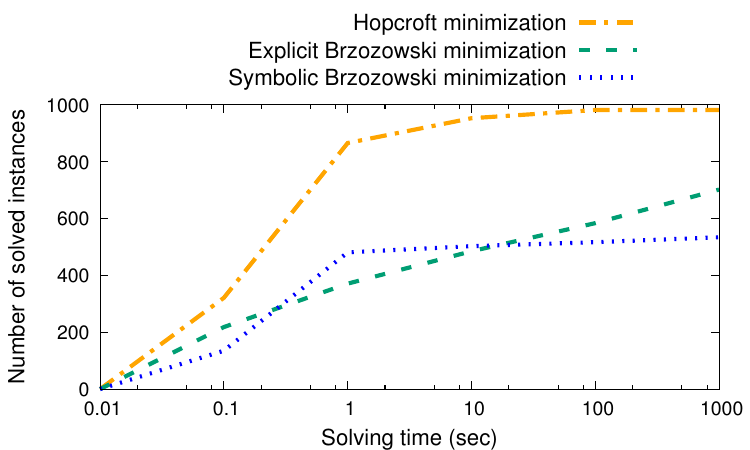}\\
    	\caption[small]{Total Running time with Hopcroft's or Brzozowski's minimization algorithms on \emph{Random} benchmarks.}
    	\label{fig:synthesis-random}
    \end{figure}
    \begin{figure}[t]
    	\centering
    	\includegraphics[width=0.7\linewidth]{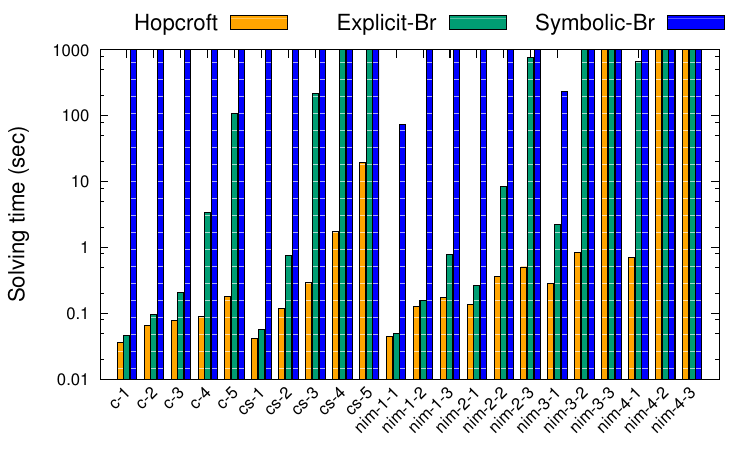}\\
    	\caption[small]{Total Running time with Hopcroft's or Brzozowski's minimization algorithms on \emph{Single-Counter}, \emph{Double-Counters}  and \emph{Nim} benchmarks.
    	}
    	\label{fig:synthesis-nonrandom}
    \end{figure} 
The results are not much different in the case of the non-random benchmarks, shown in Figure~\ref{fig:synthesis-nonrandom}. There we can see that symbolic Brzozowski's algorithm timed out for the vast majority of instances, only being able to solve the smaller instances of the \emph{Nim} family, and none of the instances of the \emph{Single-Counter} and \emph{Double-Counters} families. The explicit version performs slightly better, being able to handle some smaller instances of the \emph{Single-Counter} and \emph{Double-Counters} families. Hopcroft's algorithm, on the other hand, can solve a large number of instances within the timeout.

The results shown above allow us to answer the question of which minimization algorithm is more efficient in the context of temporal synthesis. Our data points to Hopcroft's algorithm as the better choice. It might seem surprising that Hopcroft's algorithm outperforms Brzozowski so significantly. The symbolic version in particular was expected to benefit from the fact that it is able to avoid ever having to construct the explicit \DFA for the \LTLf formula, instead constructing only the \DFA for the reverse language, which should be exponentially smaller. Yet, it fails to compete even against the explicit version. Understanding the failure of Brzozowski's algorithm requires a more in-depth investigation of the internals of the minimization procedure. We perform this analysis in the next section.

%% file: 5-analysis.tex
\section{Analysis and Discussion} \label{sec:analysis}
To understand the reasons for the results that we observed in Section~\ref{sec:experiments}, we have taken a closer look at the comparison between the minimal DFA for the formula and the DFA for the reverse language obtained in the first half of the Brzozowski's construction. 
We started by measuring the number of states of the \DFA and reverse \DFA constructed for \emph{Random} formulas. Figure~\ref{fig:dfa-random} displays a scatter plot~(in log scale) comparing the two for each instance. The blue curve indicates an exponential blowup of the number of states of the reverse DFA. In addition, the gray curve represents the points where the x-axis value is equal to the y-axis value. Thus points above the gray curve represent instances where the minimal \DFA is larger than the reverse \DFA. In several cases, the reverse \DFA is indeed smaller, as expected. We observe, however, that there is a significant number of cases where the two automata tend to have approximately the same number of states, which differs from what the theory would lead us to believe. In such cases, the benefits of using Brzozowski's construction disappear, as we can expect no advantage in constructing the minimal reverse DFA instead of directly constructing the DFA for the formula.
\begin{figure}
\centering
\begin{minipage}{.4\textwidth}
  \centering
  \includegraphics[width=\linewidth]{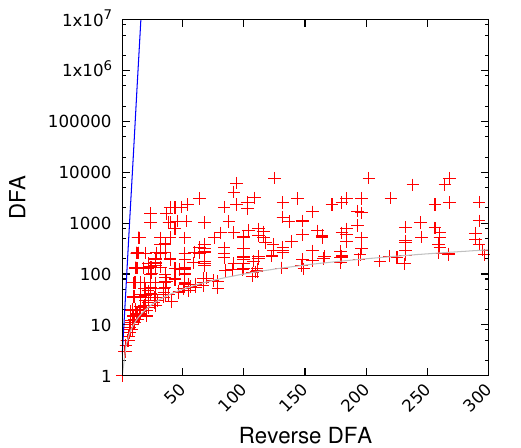}
	\caption[small]{\#states on \emph{Random} benchmarks.}
	\label{fig:dfa-random}
\end{minipage}%
\begin{minipage}{.6\textwidth}
  \centering
  \includegraphics[width=\linewidth]{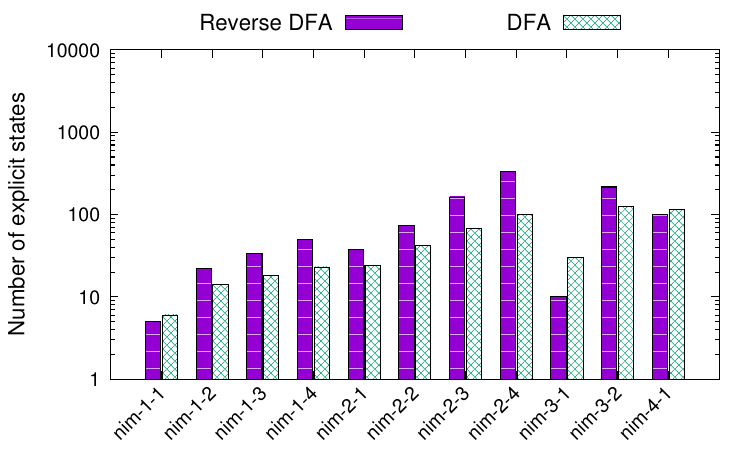}
	\caption[small]{\#states on \emph{Nim} benchmarks.}
	\label{fig:dfa-nim}
\end{minipage}
\end{figure}

Note, furthermore, that even in those cases where the reverse \DFA is indeed smaller, these points are far below the blue curve, indicating that they are not \emph{exponentially} smaller. This is a problem because, after being reversed again, the reverse DFA will go through a subset construction, which causes an exponential blowup. This is represented in the symbolic version of Brzozowski's algorithm by the number of state variables in the symbolic representation being linear in the number of the states of the reverse \DFA, instead of logarithmic in the number of states of the \DFA. Therefore, unless the reverse \DFA is exponentially smaller, the number of state variables will be larger than in the symbolic representation of the explicitly-constructed minimal DFA. This in turn impacts the performance of the winning-strategy computation.
The increase in the state space during subset construction seems to be the main reason for failures of the Brzozowski approach. This is because the minimal reverse-\DFA construction itself already takes comparable time to the minimal \DFA construction and succeeds in almost all 1000 benchmarks. Moreover, Brzozowski’s method requires further step of determinization on the reverse-\DFA, which leads to another exponential blowup that takes a lot of costs, and therefore accounts for the failure of the performance.

Interestingly, although the symbolic state-space pruning helps with the large state space during synthesis, significantly reducing the size of the BDDs representing the sets of winning states at each iteration, the computation of the set of reachable states itself ends up consuming a majority of the running time. So it turns out to not be helpful in getting the symbolic version of Brzozowski's algorithm scale.
%
%
%
\begin{figure}
	\centering
	\includegraphics[width=0.6\linewidth]{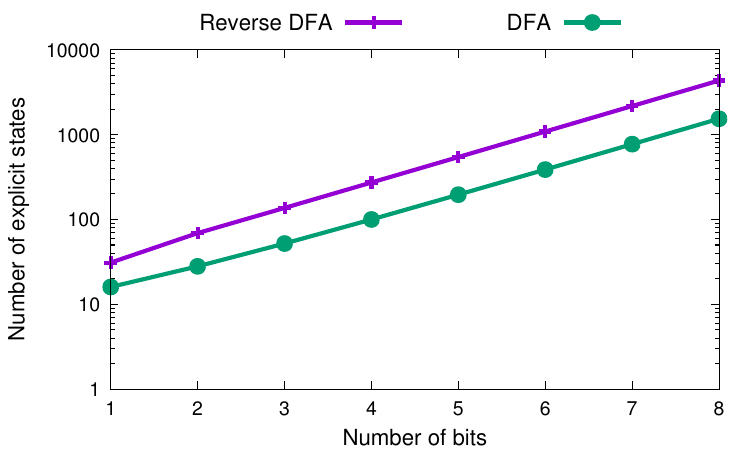}
	\caption[small]{\#states on \emph{Single-Counter} benchmarks.}
	\label{fig:dfa-scounter}
\end{figure}
%
%
The situation for the non-random benchmarks is even more extreme. Figures~\ref{fig:dfa-nim} and~\ref{fig:dfa-scounter} show the comparison of \DFA and reverse-\DFA size for the \emph{Nim} and \emph{Single-Counter} benchmarks, respectively~\footnote{The plot for the \emph{Double-Counters} benchmarks shows analogous results to the \emph{Single-Counter} benchmarks, and can be found in the appendix. The full version is on arXiv}. 
With the exception of a few of the smaller benchmarks in the \emph{Nim} family, in these instances the \DFA is actually smaller than the reverse \DFA.
The results for the counter benchmarks in particular are useful to better understand where our assumptions are violated, as we can observe the scalability of the automata in terms of the number of bits $n$ in the counter, which is proportional to the formula length. The plots (in log scale) show that the reverse \DFA grows exponentially with $n$, which is the predicted behavior. Yet, the \DFA is exponential as well, rather than doubly-exponential. 

These results highlight an important detail that is easily overlooked in the justification for the reverse DFA construction in Section~\ref{sec:reversedfa}: the theoretical lower bounds on automata sizes refer to the \emph{worst-case}. This means that even though there are formulas for which the smallest \DFA is doubly exponential, it might be that such cases occur very rarely. Our experimental results suggest that this might indeed be the case. This is consistent with previous results from~\cite{TabakovRV12}, that minimal \DFA constructed from temporal formulas are often orders of magnitude smaller in practice than the corresponding \NFA. 
Thus, even though the worst-case size of the reverse \DFA is exponentially smaller, our conclusion is that in practice the worst case is not common enough for making an approach based on constructing the reverse \DFA worthwhile. Therefore, directly constructing the minimal \DFA using Hopcroft's algorithm seems to be a better option for synthesis than employing Brzozowski's construction.

Furthermore, note that the size of the reverse DFA of \LTLf formula $\phi$ is actually the size of the minimal DFA of \pltlf formula $\phi^R$. That is to say, despite the fact that the \DFA of \pltlf is supposed to be exponentially smaller than the \DFA of \LTLf, referring to the theoretical advantage of \PLTLf over \LTLf when compiled into the corresponding DFA~\cite{GiacomoSFR20}, our data suggests that this advantage may not be common in practice. Thus, we believe that this disappearing advantage applies not only to reactive synthesis but also  to other applications of \LTLf that use compilation to \DFA. In the future, we would like to revisit this problem in other applications to confirm our conjecture.

%% file: main.bbl
\begin{thebibliography}{10}
\providecommand{\bibitemdeclare}[2]{}
\providecommand{\surnamestart}{}
\providecommand{\surnameend}{}
\providecommand{\urlprefix}{Available at }
\providecommand{\url}[1]{\texttt{#1}}
\providecommand{\href}[2]{\texttt{#2}}
\providecommand{\urlalt}[2]{\href{#1}{#2}}
\providecommand{\doi}[1]{doi:\urlalt{http://dx.doi.org/#1}{#1}}
\providecommand{\bibinfo}[2]{#2}

\bibitemdeclare{book}{AhoHU74}
\bibitem{AhoHU74}
\bibinfo{author}{Alfred~V. \surnamestart Aho\surnameend},
  \bibinfo{author}{John~E. \surnamestart Hopcroft\surnameend} \&
  \bibinfo{author}{Jeffrey~D. \surnamestart Ullman\surnameend}
  (\bibinfo{year}{1974}): \emph{\bibinfo{title}{The Design and Analysis of
  Computer Algorithms}}.
\newblock \bibinfo{publisher}{Addison-Wesley}.

\bibitemdeclare{inproceedings}{ArmoniEFKV05}
\bibitem{ArmoniEFKV05}
\bibinfo{author}{Roy \surnamestart Armoni\surnameend}, \bibinfo{author}{Sergey
  \surnamestart Egorov\surnameend}, \bibinfo{author}{Ranan \surnamestart
  Fraer\surnameend}, \bibinfo{author}{Dmitry \surnamestart
  Korchemny\surnameend} \& \bibinfo{author}{Moshe~Y. \surnamestart
  Vardi\surnameend} (\bibinfo{year}{2005}): \emph{\bibinfo{title}{{Efficient
  {LTL} compilation for SAT-based model checking}}}.
\newblock In: {\sl \bibinfo{booktitle}{{ICCAD} 2005}}, pp.
  \bibinfo{pages}{877--884}, \doi{10.1109/ICCAD.2005.1560185}.

\bibitemdeclare{article}{BacchusK98}
\bibitem{BacchusK98}
\bibinfo{author}{Fahiem \surnamestart Bacchus\surnameend} \&
  \bibinfo{author}{Froduald \surnamestart Kabanza\surnameend}
  (\bibinfo{year}{1998}): \emph{\bibinfo{title}{Planning for Temporally
  Extended Goals}}.
\newblock {\sl \bibinfo{journal}{Ann. Math. Artif. Intell.}}
  \bibinfo{volume}{22}(\bibinfo{number}{1-2}), pp. \bibinfo{pages}{5--27},
  \doi{10.1023/A:1018985923441}.

\bibitemdeclare{article}{BacchusK00}
\bibitem{BacchusK00}
\bibinfo{author}{Fahiem \surnamestart Bacchus\surnameend} \&
  \bibinfo{author}{Froduald \surnamestart Kabanza\surnameend}
  (\bibinfo{year}{2000}): \emph{\bibinfo{title}{Using temporal logics to
  express search control knowledge for planning}}.
\newblock {\sl \bibinfo{journal}{Artif. Intell.}}
  \bibinfo{volume}{116}(\bibinfo{number}{1-2}), pp. \bibinfo{pages}{123--191},
  \doi{10.1016/S0004-3702(99)00071-5}.

\bibitemdeclare{inproceedings}{BLTV}
\bibitem{BLTV}
\bibinfo{author}{Suguman \surnamestart Bansal\surnameend},
  \bibinfo{author}{Yong \surnamestart Li\surnameend}, \bibinfo{author}{Lucas~M.
  \surnamestart Tabajara\surnameend} \& \bibinfo{author}{Moshe~Y. \surnamestart
  Vardi\surnameend} (\bibinfo{year}{2020}): \emph{\bibinfo{title}{{Hybrid
  Compositional Reasoning for Reactive Synthesis from Finite-Horizon
  Specifications}}}.
\newblock In: {\sl \bibinfo{booktitle}{{AAAI}}}, pp.
  \bibinfo{pages}{9766--9774}, \doi{10.1609/aaai.v34i06.6528}.

\bibitemdeclare{inproceedings}{BertoglioLZ19}
\bibitem{BertoglioLZ19}
\bibinfo{author}{Nicola \surnamestart Bertoglio\surnameend},
  \bibinfo{author}{Gianfranco \surnamestart Lamperti\surnameend} \&
  \bibinfo{author}{Marina \surnamestart Zanella\surnameend}
  (\bibinfo{year}{2019}): \emph{\bibinfo{title}{Temporal Diagnosis of
  Discrete-Event Systems with Dual Knowledge Compilation}}.
\newblock In \bibinfo{editor}{Andreas \surnamestart Holzinger\surnameend},
  \bibinfo{editor}{Peter \surnamestart Kieseberg\surnameend},
  \bibinfo{editor}{A~Min \surnamestart Tjoa\surnameend} \&
  \bibinfo{editor}{Edgar~R. \surnamestart Weippl\surnameend}, editors: {\sl
  \bibinfo{booktitle}{{Machine Learning and Knowledge Extraction}}}, {\sl
  \bibinfo{series}{Lecture Notes in Computer Science}} \bibinfo{volume}{11713},
  \bibinfo{publisher}{Springer}, pp. \bibinfo{pages}{333--352},
  \doi{10.1007/978-3-030-29726-8\_21}.

\bibitemdeclare{inproceedings}{BienvenuFM06}
\bibitem{BienvenuFM06}
\bibinfo{author}{Meghyn \surnamestart Bienvenu\surnameend},
  \bibinfo{author}{Christian \surnamestart Fritz\surnameend} \&
  \bibinfo{author}{Sheila~A. \surnamestart McIlraith\surnameend}:
  \emph{\bibinfo{title}{Planning with Qualitative Temporal Preferences}}.
\newblock In \bibinfo{editor}{Patrick \surnamestart Doherty\surnameend},
  \bibinfo{editor}{John \surnamestart Mylopoulos\surnameend} \&
  \bibinfo{editor}{Christopher~A. \surnamestart Welty\surnameend}, editors:
  {\sl \bibinfo{booktitle}{{KR}}}.

\bibitemdeclare{inproceedings}{Bloem15}
\bibitem{Bloem15}
\bibinfo{author}{Roderick \surnamestart Bloem\surnameend}
  (\bibinfo{year}{2015}): \emph{\bibinfo{title}{{Reactive Synthesis}}}.
\newblock In: {\sl \bibinfo{booktitle}{{FMCAD}}}, \bibinfo{publisher}{{IEEE}},
  p.~\bibinfo{pages}{3}, \doi{10.1109/FMCAD.2015.7542241}.

\bibitemdeclare{article}{BloemGJPPW07}
\bibitem{BloemGJPPW07}
\bibinfo{author}{Roderick \surnamestart Bloem\surnameend},
  \bibinfo{author}{Stefan~J. \surnamestart Galler\surnameend},
  \bibinfo{author}{Barbara \surnamestart Jobstmann\surnameend},
  \bibinfo{author}{Nir \surnamestart Piterman\surnameend},
  \bibinfo{author}{Amir \surnamestart Pnueli\surnameend} \&
  \bibinfo{author}{Martin \surnamestart Weiglhofer\surnameend}
  (\bibinfo{year}{2007}): \emph{\bibinfo{title}{{Specify, Compile, Run:
  Hardware from {PSL}}}}.
\newblock {\sl \bibinfo{journal}{Electron. Notes Theor. Comput. Sci.}}
  \bibinfo{volume}{190}(\bibinfo{number}{4}), pp. \bibinfo{pages}{3--16},
  \doi{10.1016/j.entcs.2007.09.004}.

\bibitemdeclare{inproceedings}{BrafmanG19}
\bibitem{BrafmanG19}
\bibinfo{author}{Ronen~I. \surnamestart Brafman\surnameend} \&
  \bibinfo{author}{Giuseppe \surnamestart {De Giacomo}\surnameend}
  (\bibinfo{year}{2019}): \emph{\bibinfo{title}{Planning for {LTL}$_f$
  /{LDL}$_f$ Goals in Non-Markovian Fully Observable Nondeterministic
  Domains}}.
\newblock In \bibinfo{editor}{Sarit \surnamestart Kraus\surnameend}, editor:
  {\sl \bibinfo{booktitle}{{IJCAI}}}, pp. \bibinfo{pages}{1602--1608},
  \doi{10.24963/ijcai.2019/222}.

\bibitemdeclare{inproceedings}{BrafmanGP18}
\bibitem{BrafmanGP18}
\bibinfo{author}{Ronen~I. \surnamestart Brafman\surnameend},
  \bibinfo{author}{Giuseppe \surnamestart {De Giacomo}\surnameend} \&
  \bibinfo{author}{Fabio \surnamestart Patrizi\surnameend}
  (\bibinfo{year}{2018}): \emph{\bibinfo{title}{{LTL}$_f$/{LDL}$_f$
  Non-Markovian Rewards}}.
\newblock In \bibinfo{editor}{Sheila~A. \surnamestart McIlraith\surnameend} \&
  \bibinfo{editor}{Kilian~Q. \surnamestart Weinberger\surnameend}, editors:
  {\sl \bibinfo{booktitle}{{AAAI}}}, pp. \bibinfo{pages}{1771--1778}.

\bibitemdeclare{article}{Bryant92}
\bibitem{Bryant92}
\bibinfo{author}{Randal~E. \surnamestart Bryant\surnameend}
  (\bibinfo{year}{1992}): \emph{\bibinfo{title}{{Symbolic Boolean Manipulation
  with Ordered Binary-Decision Diagrams}}}.
\newblock {\sl \bibinfo{journal}{{ACM} Comput. Surv.}}
  \bibinfo{volume}{24}(\bibinfo{number}{3}), pp. \bibinfo{pages}{293--318},
  \doi{10.1145/136035.136043}.

\bibitemdeclare{inproceedings}{Brzozowski62}
\bibitem{Brzozowski62}
\bibinfo{author}{Janusz~A. \surnamestart Brzozowski\surnameend}
  (\bibinfo{year}{1962}): \emph{\bibinfo{title}{{Canonical Regular Expressions
  and Minimal State Graphs for Definite Events}}}.

\bibitemdeclare{article}{CadoliD97}
\bibitem{CadoliD97}
\bibinfo{author}{Marco \surnamestart Cadoli\surnameend} \&
  \bibinfo{author}{Francesco~M. \surnamestart Donini\surnameend}
  (\bibinfo{year}{1997}): \emph{\bibinfo{title}{A Survey on Knowledge
  Compilation}}.
\newblock {\sl \bibinfo{journal}{{AI} Commun.}}
  \bibinfo{volume}{10}(\bibinfo{number}{3-4}), pp. \bibinfo{pages}{137--150}.

\bibitemdeclare{inproceedings}{CalvaneseGV02}
\bibitem{CalvaneseGV02}
\bibinfo{author}{Diego \surnamestart Calvanese\surnameend},
  \bibinfo{author}{Giuseppe \surnamestart {De Giacomo}\surnameend} \&
  \bibinfo{author}{Moshe~Y. \surnamestart Vardi\surnameend}:
  \emph{\bibinfo{title}{Reasoning about Actions and Planning in {LTL} Action
  Theories}}.
\newblock In \bibinfo{editor}{Dieter \surnamestart Fensel\surnameend},
  \bibinfo{editor}{Fausto \surnamestart Giunchiglia\surnameend},
  \bibinfo{editor}{Deborah~L. \surnamestart McGuinness\surnameend} \&
  \bibinfo{editor}{Mary{-}Anne \surnamestart Williams\surnameend}, editors:
  {\sl \bibinfo{booktitle}{{KR}}}.

\bibitemdeclare{inproceedings}{CamachoBMM18}
\bibitem{CamachoBMM18}
\bibinfo{author}{Alberto \surnamestart Camacho\surnameend},
  \bibinfo{author}{Jorge~A. \surnamestart Baier\surnameend},
  \bibinfo{author}{Christian~J. \surnamestart Muise\surnameend} \&
  \bibinfo{author}{Sheila~A. \surnamestart McIlraith\surnameend}
  (\bibinfo{year}{2018}): \emph{\bibinfo{title}{{Finite {LTL} Synthesis as
  Planning}}}.
\newblock In: {\sl \bibinfo{booktitle}{{ICAPS}}}, pp. \bibinfo{pages}{29--38}.

\bibitemdeclare{inproceedings}{CamachoIKVM19}
\bibitem{CamachoIKVM19}
\bibinfo{author}{Alberto \surnamestart Camacho\surnameend},
  \bibinfo{author}{Rodrigo~Toro \surnamestart Icarte\surnameend},
  \bibinfo{author}{Toryn~Q. \surnamestart Klassen\surnameend},
  \bibinfo{author}{Richard~Anthony \surnamestart Valenzano\surnameend} \&
  \bibinfo{author}{Sheila~A. \surnamestart McIlraith\surnameend}
  (\bibinfo{year}{2019}): \emph{\bibinfo{title}{{LTL} and Beyond: Formal
  Languages for Reward Function Specification in Reinforcement Learning}}.
\newblock In \bibinfo{editor}{Sarit \surnamestart Kraus\surnameend}, editor:
  {\sl \bibinfo{booktitle}{{IJCAI}}}, pp. \bibinfo{pages}{6065--6073},
  \doi{10.24963/ijcai.2019/840}.

\bibitemdeclare{inproceedings}{CTMBM17}
\bibitem{CTMBM17}
\bibinfo{author}{Alberto \surnamestart Camacho\surnameend},
  \bibinfo{author}{Eleni \surnamestart Triantafillou\surnameend},
  \bibinfo{author}{Christian \surnamestart Muise\surnameend},
  \bibinfo{author}{Jorge~A. \surnamestart Baier\surnameend} \&
  \bibinfo{author}{Sheila \surnamestart McIlraith\surnameend}
  (\bibinfo{year}{2017}): \emph{\bibinfo{title}{{Non-Deterministic Planning
  with Temporally Extended Goals: {LTL} over Finite and Infinite Traces}}}.
\newblock In: {\sl \bibinfo{booktitle}{AAAI}}, pp. \bibinfo{pages}{3716--3724}.

\bibitemdeclare{article}{CKS81}
\bibitem{CKS81}
\bibinfo{author}{A.K. \surnamestart Chandra\surnameend}, \bibinfo{author}{D.C.
  \surnamestart Kozen\surnameend} \& \bibinfo{author}{L.J. \surnamestart
  Stockmeyer\surnameend} (\bibinfo{year}{1981}):
  \emph{\bibinfo{title}{Alternation}}.
\newblock {\sl \bibinfo{journal}{J. ACM}}
  \bibinfo{volume}{28}(\bibinfo{number}{1}), pp. \bibinfo{pages}{114--133},
  \doi{10.1145/322234.322243}.

\bibitemdeclare{article}{CiccioMMM17}
\bibitem{CiccioMMM17}
\bibinfo{author}{Claudio~Di \surnamestart Ciccio\surnameend},
  \bibinfo{author}{Fabrizio~Maria \surnamestart Maggi\surnameend},
  \bibinfo{author}{Marco \surnamestart Montali\surnameend} \&
  \bibinfo{author}{Jan \surnamestart Mendling\surnameend}
  (\bibinfo{year}{2017}): \emph{\bibinfo{title}{Resolving inconsistencies and
  redundancies in declarative process models}}.
\newblock {\sl \bibinfo{journal}{Inf. Syst.}} \bibinfo{volume}{64}, pp.
  \bibinfo{pages}{425--446}, \doi{10.1016/j.is.2016.09.005}.

\bibitemdeclare{article}{ConsoleTD02}
\bibitem{ConsoleTD02}
\bibinfo{author}{Luca \surnamestart Console\surnameend}, \bibinfo{author}{Paolo
  \surnamestart Terenziani\surnameend} \& \bibinfo{author}{Daniele~Theseider
  \surnamestart Dupr{\'{e}}\surnameend} (\bibinfo{year}{2002}):
  \emph{\bibinfo{title}{Local Reasoning and Knowledge Compilation for Efficient
  Temporal Abduction}}.
\newblock {\sl \bibinfo{journal}{{IEEE} Trans. Knowl. Data Eng.}}
  \bibinfo{volume}{14}(\bibinfo{number}{6}), pp. \bibinfo{pages}{1230--1248},
  \doi{10.1109/TKDE.2002.1047764}.

\bibitemdeclare{article}{DarwicheM02}
\bibitem{DarwicheM02}
\bibinfo{author}{Adnan \surnamestart Darwiche\surnameend} \&
  \bibinfo{author}{Pierre \surnamestart Marquis\surnameend}
  (\bibinfo{year}{2002}): \emph{\bibinfo{title}{A Knowledge Compilation Map}}.
\newblock {\sl \bibinfo{journal}{J. Artif. Intell. Res.}} \bibinfo{volume}{17},
  pp. \bibinfo{pages}{229--264}, \doi{10.1613/jair.989}.

\bibitemdeclare{inproceedings}{DaxEK07}
\bibitem{DaxEK07}
\bibinfo{author}{Christian \surnamestart Dax\surnameend},
  \bibinfo{author}{Jochen \surnamestart Eisinger\surnameend} \&
  \bibinfo{author}{Felix \surnamestart Klaedtke\surnameend}
  (\bibinfo{year}{2007}): \emph{\bibinfo{title}{{Mechanizing the Powerset
  Construction for Restricted Classes of \emph{omega} -Automata}}}.
\newblock In \bibinfo{editor}{Kedar~S. \surnamestart Namjoshi\surnameend},
  \bibinfo{editor}{Tomohiro \surnamestart Yoneda\surnameend},
  \bibinfo{editor}{Teruo \surnamestart Higashino\surnameend} \&
  \bibinfo{editor}{Yoshio \surnamestart Okamura\surnameend}, editors: {\sl
  \bibinfo{booktitle}{{ATVA}}}, pp. \bibinfo{pages}{223--236},
  \doi{10.1007/978-3-540-75596-8\_17}.

\bibitemdeclare{inproceedings}{GiacomoSFR20}
\bibitem{GiacomoSFR20}
\bibinfo{author}{Giuseppe \surnamestart {De Giacomo}\surnameend},
  \bibinfo{author}{Antonio \surnamestart {Di Stasio}\surnameend},
  \bibinfo{author}{Francesco \surnamestart Fuggitti\surnameend} \&
  \bibinfo{author}{Sasha \surnamestart Rubin\surnameend}
  (\bibinfo{year}{2020}): \emph{\bibinfo{title}{Pure-Past Linear Temporal and
  Dynamic Logic on Finite Traces}}.
\newblock In \bibinfo{editor}{Christian \surnamestart Bessiere\surnameend},
  editor: {\sl \bibinfo{booktitle}{{IJCAI}}}, pp. \bibinfo{pages}{4959--4965},
  \doi{10.24963/ijcai.2020/690}.

\bibitemdeclare{inproceedings}{GiacomoIFP19}
\bibitem{GiacomoIFP19}
\bibinfo{author}{Giuseppe \surnamestart {De Giacomo}\surnameend},
  \bibinfo{author}{Luca \surnamestart Iocchi\surnameend},
  \bibinfo{author}{Marco \surnamestart Favorito\surnameend} \&
  \bibinfo{author}{Fabio \surnamestart Patrizi\surnameend}
  (\bibinfo{year}{2019}): \emph{\bibinfo{title}{Foundations for Restraining
  Bolts: Reinforcement Learning with {LTL}$_f$/{LDL}$_f$ Restraining
  Specifications}}.
\newblock In: {\sl \bibinfo{booktitle}{{ICAPS}}}, pp.
  \bibinfo{pages}{128--136}.

\bibitemdeclare{inproceedings}{DegRu18}
\bibitem{DegRu18}
\bibinfo{author}{Giuseppe \surnamestart {De Giacomo}\surnameend} \&
  \bibinfo{author}{Sasha \surnamestart Rubin\surnameend}
  (\bibinfo{year}{2018}): \emph{\bibinfo{title}{{Automata-Theoretic Foundations
  of FOND Planning for {LTL}$_f$/{LDL}$_f$ Goals}}}.
\newblock In: {\sl \bibinfo{booktitle}{{IJCAI}}}, pp.
  \bibinfo{pages}{4729--4735}, \doi{10.24963/ijcai.2018/657}.

\bibitemdeclare{inproceedings}{DegVa13}
\bibitem{DegVa13}
\bibinfo{author}{Giuseppe \surnamestart {De Giacomo}\surnameend} \&
  \bibinfo{author}{Moshe~Y. \surnamestart Vardi\surnameend}
  (\bibinfo{year}{2013}): \emph{\bibinfo{title}{{Linear Temporal Logic and
  Linear Dynamic Logic on Finite Traces}}}.
\newblock In: {\sl \bibinfo{booktitle}{{IJCAI}}}, pp.
  \bibinfo{pages}{854--860}, \doi{10.5555/2540128.2540252}.

\bibitemdeclare{inproceedings}{DegVa15}
\bibitem{DegVa15}
\bibinfo{author}{Giuseppe \surnamestart {De Giacomo}\surnameend} \&
  \bibinfo{author}{Moshe~Y. \surnamestart Vardi\surnameend}
  (\bibinfo{year}{2015}): \emph{\bibinfo{title}{{Synthesis for {LTL} and {LDL}
  on Finite Traces}}}.
\newblock In: {\sl \bibinfo{booktitle}{{IJCAI}}}, pp.
  \bibinfo{pages}{1558--1564}, \doi{10.5555/2832415.2832466}.

\bibitemdeclare{inproceedings}{spot}
\bibitem{spot}
\bibinfo{author}{Alexandre \surnamestart Duret-Lutz\surnameend},
  \bibinfo{author}{Alexandre \surnamestart Lewkowicz\surnameend},
  \bibinfo{author}{Amaury \surnamestart Fauchille\surnameend},
  \bibinfo{author}{Thibaud \surnamestart Michaud\surnameend},
  \bibinfo{author}{Etienne \surnamestart Renault\surnameend} \&
  \bibinfo{author}{Laurent \surnamestart Xu\surnameend} (\bibinfo{year}{2016}):
  \emph{\bibinfo{title}{{Spot 2.0 --- A Framework for {LTL} and
  $\omega$-automata Manipulation}}}.
\newblock In: {\sl \bibinfo{booktitle}{ATVA}}, pp. \bibinfo{pages}{122--129},
  \doi{10.1007/978-3-319-46520-3\_8}.

\bibitemdeclare{inproceedings}{DuttaVT13}
\bibitem{DuttaVT13}
\bibinfo{author}{Sonali \surnamestart Dutta\surnameend},
  \bibinfo{author}{Moshe~Y. \surnamestart Vardi\surnameend} \&
  \bibinfo{author}{Deian \surnamestart Tabakov\surnameend}
  (\bibinfo{year}{2013}): \emph{\bibinfo{title}{{{CHIMP:} {A} Tool for
  Assertion-Based Dynamic Verification of SystemC Models}}}.
\newblock In: {\sl \bibinfo{booktitle}{{DIFTS@FMCAD}}}.

\bibitemdeclare{book}{FHMV1995}
\bibitem{FHMV1995}
\bibinfo{author}{Ronald \surnamestart Fagin\surnameend},
  \bibinfo{author}{Joseph~Y. \surnamestart Halpern\surnameend},
  \bibinfo{author}{Yoram \surnamestart Moses\surnameend} \&
  \bibinfo{author}{Moshe~Y. \surnamestart Vardi\surnameend}
  (\bibinfo{year}{1995}): \emph{\bibinfo{title}{Reasoning About Knowledge}}.
\newblock \bibinfo{publisher}{{MIT} Press},
  \doi{10.7551/mitpress/5803.001.0001}.

\bibitemdeclare{incollection}{FisherW05}
\bibitem{FisherW05}
\bibinfo{author}{Michael \surnamestart Fisher\surnameend} \&
  \bibinfo{author}{Michael~J. \surnamestart Wooldridge\surnameend}
  (\bibinfo{year}{2005}): \emph{\bibinfo{title}{Temporal Reasoning in
  Agent-Based Systems}}.
\newblock In \bibinfo{editor}{Michael \surnamestart Fisher\surnameend},
  \bibinfo{editor}{Dov~M. \surnamestart Gabbay\surnameend} \&
  \bibinfo{editor}{Llu{\'{\i}}s \surnamestart Vila\surnameend}, editors: {\sl
  \bibinfo{booktitle}{Handbook of Temporal Reasoning in Artificial
  Intelligence}}, {\sl \bibinfo{series}{Foundations of Artificial
  Intelligence}}~\bibinfo{volume}{1}, \bibinfo{publisher}{Elsevier}, pp.
  \bibinfo{pages}{469--495}, \doi{10.1016/S1574-6526(05)80017-3}.

\bibitemdeclare{inproceedings}{FogartyKVW13}
\bibitem{FogartyKVW13}
\bibinfo{author}{Seth \surnamestart Fogarty\surnameend}, \bibinfo{author}{Orna
  \surnamestart Kupferman\surnameend}, \bibinfo{author}{Moshe~Y. \surnamestart
  Vardi\surnameend} \& \bibinfo{author}{Thomas \surnamestart Wilke\surnameend}
  (\bibinfo{year}{2013}): \emph{\bibinfo{title}{{Profile Trees for B{\"{u}}chi
  Word Automata, with Application to Determinization}}}.
\newblock In: {\sl \bibinfo{booktitle}{{GandALF}}}, pp.
  \bibinfo{pages}{107--121}, \doi{10.4204/EPTCS.119.11}.

\bibitemdeclare{inproceedings}{DrLu.cav16}
\bibitem{DrLu.cav16}
\bibinfo{author}{Dror \surnamestart Fried\surnameend},
  \bibinfo{author}{Lucas~M. \surnamestart Tabajara\surnameend} \&
  \bibinfo{author}{Moshe~Y. \surnamestart Vardi\surnameend}
  (\bibinfo{year}{2016}): \emph{\bibinfo{title}{{BDD-Based Boolean Functional
  Synthesis}}}.
\newblock In: {\sl \bibinfo{booktitle}{{CAV}}}, pp. \bibinfo{pages}{402--421},
  \doi{10.1007/978-3-319-41540-6\_22}.

\bibitemdeclare{inproceedings}{GiacomoMGMM14}
\bibitem{GiacomoMGMM14}
\bibinfo{author}{Giuseppe~De \surnamestart Giacomo\surnameend},
  \bibinfo{author}{Riccardo~De \surnamestart Masellis\surnameend},
  \bibinfo{author}{Marco \surnamestart Grasso\surnameend},
  \bibinfo{author}{Fabrizio~Maria \surnamestart Maggi\surnameend} \&
  \bibinfo{author}{Marco \surnamestart Montali\surnameend}
  (\bibinfo{year}{2014}): \emph{\bibinfo{title}{Monitoring Business
  Metaconstraints Based on {LTL} and {LDL} for Finite Traces}}.
\newblock In \bibinfo{editor}{Shazia~Wasim \surnamestart Sadiq\surnameend},
  \bibinfo{editor}{Pnina \surnamestart Soffer\surnameend} \&
  \bibinfo{editor}{Hagen \surnamestart V{\"{o}}lzer\surnameend}, editors: {\sl
  \bibinfo{booktitle}{{BPM}}}, {\sl \bibinfo{series}{Lecture Notes in Computer
  Science}} \bibinfo{volume}{8659}, pp. \bibinfo{pages}{1--17},
  \doi{10.1007/978-3-319-10172-9\_1}.

\bibitemdeclare{inproceedings}{GiacomoV99}
\bibitem{GiacomoV99}
\bibinfo{author}{Giuseppe~De \surnamestart Giacomo\surnameend} \&
  \bibinfo{author}{Moshe~Y. \surnamestart Vardi\surnameend}
  (\bibinfo{year}{1999}): \emph{\bibinfo{title}{Automata-Theoretic Approach to
  Planning for Temporally Extended Goals}}.
\newblock In \bibinfo{editor}{Susanne \surnamestart Biundo\surnameend} \&
  \bibinfo{editor}{Maria \surnamestart Fox\surnameend}, editors: {\sl
  \bibinfo{booktitle}{{ECP}}}, {\sl \bibinfo{series}{Lecture Notes in Computer
  Science}} \bibinfo{volume}{1809}, \bibinfo{publisher}{Springer}, pp.
  \bibinfo{pages}{226--238}, \doi{10.1007/10720246\_18}.

\bibitemdeclare{inproceedings}{HeWKV19}
\bibitem{HeWKV19}
\bibinfo{author}{Keliang \surnamestart He\surnameend},
  \bibinfo{author}{Andrew~M. \surnamestart Wells\surnameend},
  \bibinfo{author}{Lydia~E. \surnamestart Kavraki\surnameend} \&
  \bibinfo{author}{Moshe~Y. \surnamestart Vardi\surnameend}
  (\bibinfo{year}{2019}): \emph{\bibinfo{title}{{Efficient Symbolic Reactive
  Synthesis for Finite-Horizon Tasks}}}.
\newblock In: {\sl \bibinfo{booktitle}{{ICRA}}}, pp.
  \bibinfo{pages}{8993--8999}, \doi{10.1109/ICRA.2019.8794170}.

\bibitemdeclare{inproceedings}{Mona}
\bibitem{Mona}
\bibinfo{author}{Jesper~G. \surnamestart Henriksen\surnameend},
  \bibinfo{author}{Jakob~L. \surnamestart Jensen\surnameend},
  \bibinfo{author}{Michael~E. \surnamestart J{\o}rgensen\surnameend},
  \bibinfo{author}{Nils \surnamestart Klarlund\surnameend},
  \bibinfo{author}{Robert \surnamestart Paige\surnameend},
  \bibinfo{author}{Theis \surnamestart Rauhe\surnameend} \&
  \bibinfo{author}{Anders \surnamestart Sandholm\surnameend}
  (\bibinfo{year}{1995}): \emph{\bibinfo{title}{{Mona: Monadic Second-order
  Logic in Practice}}}.
\newblock In: {\sl \bibinfo{booktitle}{TACAS}}, pp. \bibinfo{pages}{89--110},
  \doi{10.1007/3-540-60630-0\_5}.

\bibitemdeclare{techreport}{Hopcroft71}
\bibitem{Hopcroft71}
\bibinfo{author}{John~E. \surnamestart Hopcroft\surnameend}
  (\bibinfo{year}{1971}): \emph{\bibinfo{title}{{An n Log n Algorithm for
  Minimizing States in a Finite Automaton}}}.
\newblock \bibinfo{type}{Technical Report}, \bibinfo{address}{Stanford, CA,
  USA}.

\bibitemdeclare{inproceedings}{Kupferman12}
\bibitem{Kupferman12}
\bibinfo{author}{Orna \surnamestart Kupferman\surnameend}
  (\bibinfo{year}{2012}): \emph{\bibinfo{title}{{Recent Challenges and Ideas in
  Temporal Synthesis}}}.
\newblock In: {\sl \bibinfo{booktitle}{{SOFSEM}}}, pp. \bibinfo{pages}{88--98},
  \doi{10.1007/978-3-642-27660-6\_8}.

\bibitemdeclare{inproceedings}{KupfermanV98}
\bibitem{KupfermanV98}
\bibinfo{author}{Orna \surnamestart Kupferman\surnameend} \&
  \bibinfo{author}{Moshe~Y. \surnamestart Vardi\surnameend}
  (\bibinfo{year}{1998}): \emph{\bibinfo{title}{{Freedom, Weakness, and
  Determinism: From Linear-Time to Branching-Time}}}.
\newblock In: {\sl \bibinfo{booktitle}{{LICS}}}, pp. \bibinfo{pages}{81--92},
  \doi{10.1109/LICS.1998.705645}.

\bibitemdeclare{article}{KupfermanVa01}
\bibitem{KupfermanVa01}
\bibinfo{author}{Orna \surnamestart Kupferman\surnameend} \&
  \bibinfo{author}{Moshe~Y. \surnamestart Vardi\surnameend}
  (\bibinfo{year}{2001}): \emph{\bibinfo{title}{{Model Checking of Safety
  Properties}}}.
\newblock {\sl \bibinfo{journal}{Formal Methods in System Design}}
  \bibinfo{volume}{19}(\bibinfo{number}{3}), pp. \bibinfo{pages}{291--314},
  \doi{10.1023/A:1011254632723}.

\bibitemdeclare{inproceedings}{KupfermanV05}
\bibitem{KupfermanV05}
\bibinfo{author}{Orna \surnamestart Kupferman\surnameend} \&
  \bibinfo{author}{Moshe~Y. \surnamestart Vardi\surnameend}
  (\bibinfo{year}{2005}): \emph{\bibinfo{title}{{Safraless Decision
  Procedures}}}.
\newblock In: {\sl \bibinfo{booktitle}{{FOCS}}}, pp. \bibinfo{pages}{531--542},
  \doi{10.1109/SFCS.2005.66}.

\bibitemdeclare{inproceedings}{LichtensteinPZ85}
\bibitem{LichtensteinPZ85}
\bibinfo{author}{Orna \surnamestart Lichtenstein\surnameend},
  \bibinfo{author}{Amir \surnamestart Pnueli\surnameend} \&
  \bibinfo{author}{Lenore~D. \surnamestart Zuck\surnameend}
  (\bibinfo{year}{1985}): \emph{\bibinfo{title}{{The Glory of the Past}}}.
\newblock In: {\sl \bibinfo{booktitle}{Logics of Programs}}, pp.
  \bibinfo{pages}{196--218}, \doi{10.1007/3-540-15648-8\_16}.

\bibitemdeclare{inproceedings}{MorgensternS08}
\bibitem{MorgensternS08}
\bibinfo{author}{Andreas \surnamestart Morgenstern\surnameend} \&
  \bibinfo{author}{Klaus \surnamestart Schneider\surnameend}
  (\bibinfo{year}{2008}): \emph{\bibinfo{title}{{From {LTL} to Symbolically
  Represented Deterministic Automata}}}.
\newblock In \bibinfo{editor}{Francesco \surnamestart Logozzo\surnameend},
  \bibinfo{editor}{Doron~A. \surnamestart Peled\surnameend} \&
  \bibinfo{editor}{Lenore~D. \surnamestart Zuck\surnameend}, editors: {\sl
  \bibinfo{booktitle}{{VMCAI}}}, pp. \bibinfo{pages}{279--293},
  \doi{10.1007/978-3-540-78163-9\_24}.

\bibitemdeclare{inproceedings}{PesicSA07}
\bibitem{PesicSA07}
\bibinfo{author}{Maja \surnamestart Pesic\surnameend}, \bibinfo{author}{Helen
  \surnamestart Schonenberg\surnameend} \& \bibinfo{author}{Wil M.~P.
  \surnamestart van~der Aalst\surnameend} (\bibinfo{year}{2007}):
  \emph{\bibinfo{title}{{DECLARE:} Full Support for Loosely-Structured
  Processes}}.
\newblock In: {\sl \bibinfo{booktitle}{{(EDOC}}}, pp.
  \bibinfo{pages}{287--300}, \doi{10.1109/EDOC.2007.14}.

\bibitemdeclare{inproceedings}{Pin87}
\bibitem{Pin87}
\bibinfo{author}{Jean{-}Eric \surnamestart Pin\surnameend}
  (\bibinfo{year}{1987}): \emph{\bibinfo{title}{{On the Language Accepted by
  Finite Reversible Automata}}}.
\newblock In \bibinfo{editor}{Thomas \surnamestart Ottmann\surnameend}, editor:
  {\sl \bibinfo{booktitle}{{ICALP}}}, pp. \bibinfo{pages}{237--249},
  \doi{10.1007/3-540-18088-5\_19}.

\bibitemdeclare{inproceedings}{Pnu77}
\bibitem{Pnu77}
\bibinfo{author}{Amir \surnamestart Pnueli\surnameend} (\bibinfo{year}{1977}):
  \emph{\bibinfo{title}{{The temporal logic of programs}}}.
\newblock pp. \bibinfo{pages}{46--57}, \doi{10.1109/SFCS.1977.32}.

\bibitemdeclare{inproceedings}{PnueliR89}
\bibitem{PnueliR89}
\bibinfo{author}{Amir \surnamestart Pnueli\surnameend} \& \bibinfo{author}{Roni
  \surnamestart Rosner\surnameend} (\bibinfo{year}{1989}):
  \emph{\bibinfo{title}{{On the Synthesis of a Reactive Module}}}.
\newblock In: {\sl \bibinfo{booktitle}{{POPL}}}, pp. \bibinfo{pages}{179--190},
  \doi{10.1145/75277.75293}.

\bibitemdeclare{article}{RS59}
\bibitem{RS59}
\bibinfo{author}{M.O. \surnamestart Rabin\surnameend} \&
  \bibinfo{author}{D.~\surnamestart Scott\surnameend} (\bibinfo{year}{1959}):
  \emph{\bibinfo{title}{Finite automata and their decision problems}}.
\newblock {\sl \bibinfo{journal}{IBM Journal of Research and Development}}
  \bibinfo{volume}{3}, pp. \bibinfo{pages}{115--125}, \doi{10.1147/rd.32.0114}.

\bibitemdeclare{article}{cudd}
\bibitem{cudd}
\bibinfo{author}{Fabio \surnamestart Somenzi\surnameend}
  (\bibinfo{year}{2016}): \emph{\bibinfo{title}{{CUDD: CU Decision Diagram
  Package 3.0.0. Universiy of Colorado at Boulder}}}.

\bibitemdeclare{inproceedings}{TabajaraV19}
\bibitem{TabajaraV19}
\bibinfo{author}{Lucas~Martinelli \surnamestart Tabajara\surnameend} \&
  \bibinfo{author}{Moshe~Y. \surnamestart Vardi\surnameend}
  (\bibinfo{year}{2019}): \emph{\bibinfo{title}{{Partitioning Techniques in
  LTL$_f$ Synthesis}}}.
\newblock In: {\sl \bibinfo{booktitle}{{IJCAI}}}, pp.
  \bibinfo{pages}{5599--5606}, \doi{10.24963/ijcai.2019/777}.

\bibitemdeclare{article}{TabakovRV12}
\bibitem{TabakovRV12}
\bibinfo{author}{Deian \surnamestart Tabakov\surnameend},
  \bibinfo{author}{Kristin~Y. \surnamestart Rozier\surnameend} \&
  \bibinfo{author}{Moshe~Y. \surnamestart Vardi\surnameend}
  (\bibinfo{year}{2012}): \emph{\bibinfo{title}{{Optimized temporal monitors
  for {SystemC}}}}.
\newblock {\sl \bibinfo{journal}{Formal Methods in System Design}}
  \bibinfo{volume}{41}(\bibinfo{number}{3}), pp. \bibinfo{pages}{236--268},
  \doi{10.1007/s10703-011-0139-8}.

\bibitemdeclare{inproceedings}{TabakovV05}
\bibitem{TabakovV05}
\bibinfo{author}{Deian \surnamestart Tabakov\surnameend} \&
  \bibinfo{author}{Moshe~Y. \surnamestart Vardi\surnameend}
  (\bibinfo{year}{2005}): \emph{\bibinfo{title}{{Experimental Evaluation of
  Classical Automata Constructions}}}.
\newblock In: {\sl \bibinfo{booktitle}{{LPAR}}}, pp. \bibinfo{pages}{396--411},
  \doi{10.1007/11591191\_28}.

\bibitemdeclare{inproceedings}{AndrewGandalf}
\bibitem{AndrewGandalf}
\bibinfo{author}{Andrew~M. \surnamestart Wells\surnameend},
  \bibinfo{author}{Morteza \surnamestart Lahijanian\surnameend},
  \bibinfo{author}{Lydia~E. \surnamestart Kavraki\surnameend} \&
  \bibinfo{author}{Moshe~Y. \surnamestart Vardi\surnameend}
  (\bibinfo{year}{2020}): \emph{\bibinfo{title}{LTL$_f$ Synthesis on
  Probabilistic Systems}}.
\newblock In \bibinfo{editor}{Jean{-}Fran{\c{c}}ois \surnamestart
  Raskin\surnameend} \& \bibinfo{editor}{Davide \surnamestart
  Bresolin\surnameend}, editors: {\sl \bibinfo{booktitle}{{GandALF}}}, {\sl
  \bibinfo{series}{{EPTCS}}} \bibinfo{volume}{326}, pp.
  \bibinfo{pages}{166--181}, \doi{10.4204/EPTCS.326.11}.

\bibitemdeclare{article}{abs-2101-11981}
\bibitem{abs-2101-11981}
\bibinfo{author}{Yaqi \surnamestart Xie\surnameend}, \bibinfo{author}{Fan
  \surnamestart Zhou\surnameend} \& \bibinfo{author}{Harold \surnamestart
  Soh\surnameend} (\bibinfo{year}{2021}): \emph{\bibinfo{title}{Embedding
  Symbolic Temporal Knowledge into Deep Sequential Models}}.
\newblock {\sl \bibinfo{journal}{CoRR}} \bibinfo{volume}{abs/2101.11981}.

\bibitemdeclare{inproceedings}{ZhuGPV20}
\bibitem{ZhuGPV20}
\bibinfo{author}{Shufang \surnamestart Zhu\surnameend},
  \bibinfo{author}{Giuseppe~De \surnamestart Giacomo\surnameend},
  \bibinfo{author}{Geguang \surnamestart Pu\surnameend} \&
  \bibinfo{author}{Moshe~Y. \surnamestart Vardi\surnameend}
  (\bibinfo{year}{2020}): \emph{\bibinfo{title}{{LTL}$_f$ Synthesis with
  Fairness and Stability Assumptions}}.
\newblock In: {\sl \bibinfo{booktitle}{{AAAI}}}, pp.
  \bibinfo{pages}{3088--3095}, \doi{10.1609/aaai.v34i03.5704}.

\bibitemdeclare{inproceedings}{ZhuPV19}
\bibitem{ZhuPV19}
\bibinfo{author}{Shufang \surnamestart Zhu\surnameend},
  \bibinfo{author}{Geguang \surnamestart Pu\surnameend} \&
  \bibinfo{author}{Moshe~Y. \surnamestart Vardi\surnameend}
  (\bibinfo{year}{2019}): \emph{\bibinfo{title}{{First-Order vs. Second-Order
  Encodings for LTL$_f$-to-Automata Translation}}}.
\newblock In: {\sl \bibinfo{booktitle}{{TAMC}}}, pp. \bibinfo{pages}{684--705},
  \doi{10.1007/978-3-030-14812-6\_43}.

\bibitemdeclare{inproceedings}{ZhuTLPV17}
\bibitem{ZhuTLPV17}
\bibinfo{author}{Shufang \surnamestart Zhu\surnameend},
  \bibinfo{author}{Lucas~M. \surnamestart Tabajara\surnameend},
  \bibinfo{author}{Jianwen \surnamestart Li\surnameend},
  \bibinfo{author}{Geguang \surnamestart Pu\surnameend} \&
  \bibinfo{author}{Moshe~Y. \surnamestart Vardi\surnameend}
  (\bibinfo{year}{2017}): \emph{\bibinfo{title}{{A Symbolic Approach to Safety
  LTL Synthesis}}}.
\newblock In: {\sl \bibinfo{booktitle}{HVC}}, pp. \bibinfo{pages}{147--162},
  \doi{10.1007/978-3-319-70389-3\_10}.

\bibitemdeclare{inproceedings}{ZTLPV17}
\bibitem{ZTLPV17}
\bibinfo{author}{Shufang \surnamestart Zhu\surnameend},
  \bibinfo{author}{Lucas~M. \surnamestart Tabajara\surnameend},
  \bibinfo{author}{Jianwen \surnamestart Li\surnameend},
  \bibinfo{author}{Geguang \surnamestart Pu\surnameend} \&
  \bibinfo{author}{Moshe~Y. \surnamestart Vardi\surnameend}
  (\bibinfo{year}{2017}): \emph{\bibinfo{title}{{Symbolic {LTL$_f$}
  Synthesis}}}.
\newblock In: {\sl \bibinfo{booktitle}{{IJCAI}}}, pp.
  \bibinfo{pages}{1362--1369}, \doi{10.24963/ijcai.2017/189}.

\end{thebibliography}
